\documentclass[submission]{eptcs}
\providecommand{\event}{TERMGRAPH 2011} 
\usepackage{breakurl}   
\usepackage[mac]{inputenc}
\usepackage{bm}
\usepackage{amssymb}
\usepackage{amsmath}
\usepackage{amsfonts}
\usepackage{amssymb}
\usepackage{subfigure}
\usepackage{graphicx}
\usepackage{subfigure}
\usepackage[usenames,dvipsnames]{color}
\usepackage{listings} 
\usepackage{verbatim}
\usepackage{color}
\usepackage[all]{xy}
\usepackage{wrapfig}

\newcommand{\com}[1]{}

\newtheorem{definition}{Definition}

\newtheorem{lemma}{Lemma}
\newtheorem{theorem}{Theorem}
\newenvironment{proof}{\noindent\textit{Proof. }}{\hfill $\square$\\}

\begin{document}

\title{Rule-based transformations for geometric modelling}

\author{Thomas Bellet
\institute{University of Poitiers, XLIM-SIC CNRS, France}
\email{thomas.bellet@univ-poitiers.fr}
\and
Agnès Arnould
\institute{University of Poitiers, XLIM-SIC CNRS, France}
\email{agnes.arnould@univ-poitiers.fr}
\and
Pascale Le Gall
\institute{Ecole Centrale Paris, MAS, France}
\email{pascale.legall@ecp.fr}
}
\def\titlerunning{Rule-based transformations for geometric modelling}
\def\authorrunning{T. Bellet, A. Arnould \& P. Le Gall}


\maketitle


\begin{abstract}

The context of this paper is the use of formal methods for topology-based geometric modelling. Topology-based geometric modelling deals with objects of various dimensions and shapes. Usually, objects are defined by a graph-based topological data structure and by an embedding that associates each topological element (vertex, edge, face, etc.) with relevant data as their geometric shape (position, curve, surface, etc.) or application dedicated data (e.g. molecule concentration level in a biological context). We propose to define topology-based geometric objects as labelled graphs. The arc labelling defines the topological structure of the object whose topological consistency is then ensured by labelling constraints. Nodes have as many labels as there are different data kinds in the embedding.  Labelling constraints ensure then that the embedding is consistent with the topological structure. Thus, topology-based geometric objects constitute a particular subclass of a category of labelled graphs in which nodes have multiple labels.

We previously introduced a formal approach of topological modelling based on graph transformation rules.  Topological operations, that only modify the topological structure of objects, can be defined such that the topological consistency of constructed objects is ensured with syntactic conditions on rules.  In this paper, we follow the same approach in order to deal with geometric operations, that can modify both the topological structure and the embedding. Thus, we define syntactic conditions on rules to ensure the consistency of the embedding during transformations.

\end{abstract}

\section*{Introduction}

Topology-based geometric modelling deals with the manipulation (construction, modification, \ldots) of objects that are subdivided according to their topological structure. The topological structure is the cell subdivision (vertices, edges, faces, volumes) of objects and the adjacency relations between these cells. Among the existing topological models, we choose in this paper the model of generalized maps \cite{Lienhardt89,Lienhardt94}, also called G-maps. 
The topological structure of G-maps can be represented by a graph where edges indicate which nodes are neighbours and where edge labels indicate what kind of neighbouring is concerned ({\em e.g.} connection between faces or between volumes). This graph must satisfy some constraints on the arc labelling to ensure the topological consistency of the topological structure. For example, while their shapes are different, the three objects of Fig.~\ref{fig:diff_ebd_topo} have the same topological structure: a closed face that contains four edges and four vertices. 

\begin{figure}[h]
    \centering
        \includegraphics[width=70mm]{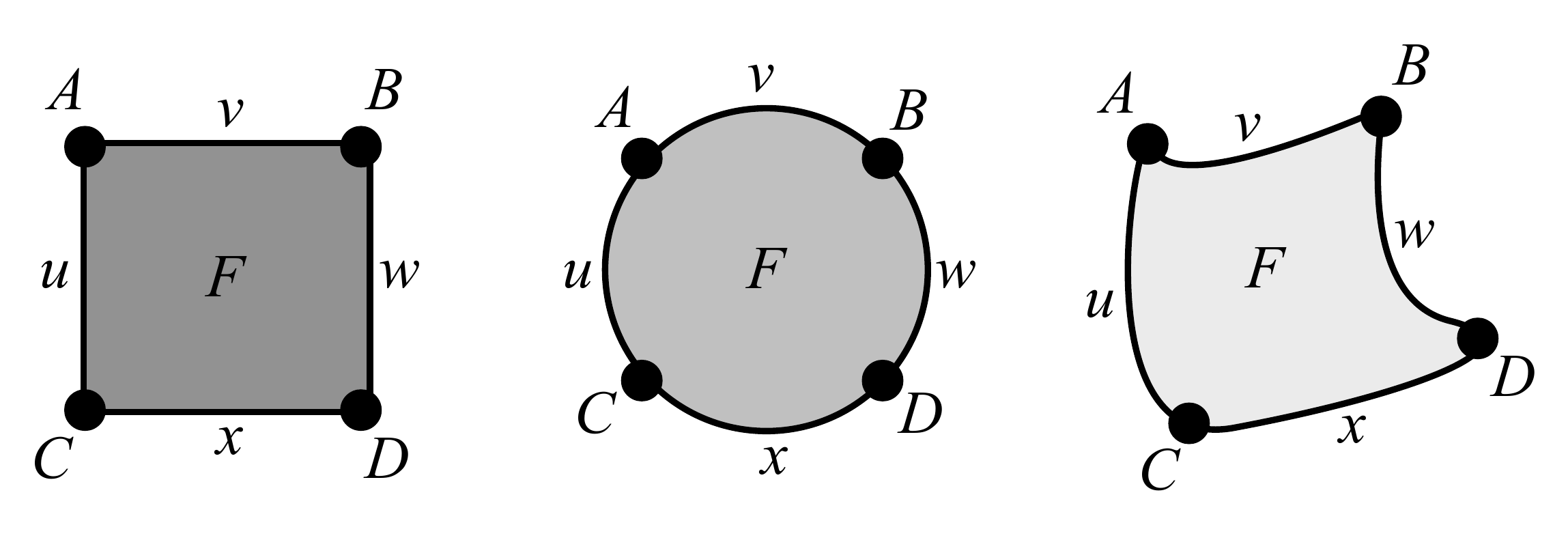}
        \vspace{-0.3cm}
    \caption{Three objects with a same topological structure}
    \label{fig:diff_ebd_topo}
\end{figure}

In addition to the topological structure, objects are defined by an embedding that includes all other kinds of information attached to the topological cells of the object. An evident example of embedding is given by the kinds of information needed to capture the shape of objects. For the objects of  Fig.~\ref{fig:diff_ebd_topo}, we assume that the associated embedding contains three elements:
\begin{itemize}
\item geometric points (defined by 2 dimension coordinates in the case of plane objects) that are attached to topological vertices;
\item curves that are attached to edges;
\item colors that are attached to faces.
\end{itemize}
While the topological structure is represented by the arc labels, the different elements of the object embedding can be represented by node labels. Intuitively, a node is labelled by all the embedding elements that are attached to its adjacent cells (vertex, edges, faces, volumes). Let us point out that while the embedding generally contains classical geometric data describing the shape of the objects (e.g. points, curves, surfaces, etc.), the embedding also contains specific data that depend on the targeted application ({\em e.g.} molecule concentration for biology, rock density for geology, material for architecture, etc.). Actually, nodes have as many labels as there are different kinds of data in the application-oriented definition of the embedding.  This fact explains that in a first step, we provide in Section~\ref{sec:1-multi_graph} a category of graphs whose nodes can carry multiple labels. This category is defined as a direct extension of the category of partially labelled graphs as defined in \cite{Habel-Plump02}. We then define in Section~\ref{sec:2-gmap} our embedded topological model as particular graphs of this  category.  Graphs that represent embedded G-maps have to satisfy constraints to ensure both the topological consistency and the embedding consistency.

To define operations on objects represented as embedded G-maps, we choose the graph transformations and more precisely the so-called double-pushout approach~\cite{ehrig2006}.  In a previous work~\cite{PACL2008_2394}, we defined a rule-based language dedicated to topology-based modelling. The first interest of this language is that we defined syntactic conditions on rules to ensure by construction, that the application of a rule to a G-map produces a G-map.
In other words, objects resulting from applications of well-formed rules on G-maps are systematically well-formed topological objects, that is objects satisfying the topological consistency constraints. In this paper, we define a similar framework for geometric operations that can modify both the topological structure and the embedding. As we need to  change labels of nodes and arcs during transformations, we based our work on rules of \cite{Habel-Plump02} that allow to rename labels. In Section~\ref{sec:3-gmap_rule}, similarly to the topological conditions introduced in~\cite{PACL2008_2394}, we define syntactic conditions on rules that ensure the preservation of the embedding constraints when rules are applied to embedded G-maps.  In Section~\ref{sec:4-rule_scheme}, we provide G-map rule schemes that allow to define generic geometric operations. Actually, rule schemes contain expressions on variables to allow to compute the embedding of the resulting objects. Finally, we provide syntactic conditions on rule schemes to ensure the preservation of embedding consistency by rule application.

\section{Transformation rules for  $I$-labelled graphs}
\label{sec:1-multi_graph}

\subsection{Category of  $I$-labelled graphs}
In this section, we define the category of  $I$-labelled graphs as an extension of the one of partially labelled graphs defined in~\cite{Habel-Plump02}. While in~\cite{Habel-Plump02} nodes have at most one label, in our case, nodes can have at most $|I|$ labels where $I$ is a chosen set of indexes.

\begin{definition}[ $I$-labelled graph]
\label{def:graphe}
Let $(\mathcal{C}_{V,i})_{i \in I}$ be a family of node label sets and $\mathcal{C}_E$ be an arc label set.
A {\em  $I$-labelled graph} $G^I = (V, E, s, t, (l_{V,i})_{i \in I}, l_E)$ upon $(\mathcal{C}_{V,i})_{i \in I}$ and $\mathcal{C}_E$ is defined as:
\begin{itemize}
\item a set $V$ of nodes;
\item a set $E$ of arcs;
\item two functions \textup{source} $s : E \rightarrow V$ and \textup{target} $t : E \rightarrow V$. For $e \in E$, $s(e)$ and $t(e)$ are respectively the \textup{source node}  and the \textup{target node} of $e$;
\item a family of partial functions\footnote{Given $X$ and $Y$ two sets, a partial function $f$ from $X$ to $Y$ is a total function $f: X' \rightarrow Y$, from $X'$ a subset of $X$. $X'$ is called the domain of $f$, and is denoted by $Dom(f)$. For $x \in X - Dom(f)$, we say that $f(x)$ is undefined, and write $f(x) = \bot$. We also note $\bot : X \rightarrow Y$ the function totally undefined, that is $Dom(\bot) = \emptyset$.}  $(l_{V,i} : V \rightarrow \mathcal{C}_{V,i})_{i \in I}$  that label nodes. For $v \in V$, when it exists, $l_{V,i}(v)$ is called the $i$-label of $v$ ;
\item a partial function $l_E : E \rightarrow \mathcal{C}_E$ that labels arcs.
\end{itemize}
\end{definition}

For a graph $G^I = (V, E, s, t, (l_{V,i})_{i \in I}, l_E)$, elements of the tuple can be indexed by $G$ to make explicit the graph name: $V_G$ for $V$ for example. The above definition is a natural extension of partially labelled graphs of \cite{Habel-Plump02}. Indeed, instead of a unique partial function $l_V$ that labels nodes, we consider an $I$-indexed family $(l_{V,i})_{i \in I}$ of partial labelling functions\footnote{To better fit with the frame of  \cite{Habel-Plump02}, one would think to label nodes by a unique label made of a Cartesian product, instead of having a family of labelling functions. But, such an approach would not allow us to have the possibility of labelling a node simultaneously by a defined $i$-label  and by an undefined $i'$-label  for $i$ and $i'$ indexes of $I$.}.  By extending the definition given in \cite{Habel-Plump02} for a unique node labelling function, an $I$-labelled morphism $g : G^I \rightarrow G'^I$ between  $I$-labelled graphs $G^I$ and $G'^I$ is defined by two functions $g_V : V_G \rightarrow V_{G'}$ and $g_E : E_G \rightarrow E_{G'}$ preserving sources, targets and labels : $s_{G'} \circ g_E = g_V \circ s_G$, $t_{G'} \circ g_{E} = g_V \circ t_G$, for all $x$ in $Dom(l_{G,E})$ , $l_{G',E}(g_E(x)) = l_{G,E}(x)$ and lastly, for all $i$ in $I$, for all $x$ in $Dom(l_{G,V,i})$, $l_{G',V,i}(g_V(x)) = l_{G,V,i}(x)$. Thus, the only difference with \cite{Habel-Plump02} is that for  $I$-labelled graphs, $I$-labelled morphisms have more labels to preserve. An $I$-labelled morphism $g : G^I \rightarrow G'^I$ is an inclusion if $\forall x \in E_G, g_E(x) =x$ and $\forall x \in V_G, g_V(x)=x$. Such an inclusion is then denoted as $g : G^I \hookrightarrow G'^I$.  $I$-labelled graphs and $I$-labelled morphisms  constitute a category, where morphism composition is defined componentwise as function composition.

For any  partially labelled graph $G = (V,E,s,t,l_V,l_E)$, we call the base of $G$ the  partially labelled graph defined as $(V, E, s, t, \bot, l_E)$ whose node labelling is totally undefined and denote it by~$G_{\bot}$.

We say that two morphisms $g : G \rightarrow G'$ and $h : H \rightarrow H'$ between  partially labelled graphs have the same base  if $G_{\bot} = H_{\bot}$, $G'_{\bot} = H'_{\bot} $ and $g_V = h_V$, $g_E = h_E$. We note $g_{\bot} : G_{\bot} \rightarrow H_{\bot}$ the derived morphism defined by  ${g_{\bot}}_ E = g_E$ and ${g_{\bot}}_ V = g_V$.

We respectively note ${\cal G}$, ${\cal G}^I$ and ${\cal G}^{\bot}$ the category of partially labelled graphs (as defined in \cite{Habel-Plump02}),  $I$-labelled graphs and bases of  partially labelled graphs (that is, graphs whose node labelling is the function totally undefined).

For a  $I$-labelled graph $G^I= (V, E, s, t, (l_{V,i})_{i \in I}, l_E)$, for an index $i \in I$, the projection  $proj_i(G^I)$, also called the $i$-component, is defined as the  partially labelled graph $ (V, E, s, t, l_{V,i}, l_E)$ according to \cite{Habel-Plump02}.  Similarly, for an $I$-labelled morphism $g : G^I \rightarrow G'^I$, we call $proj_i(g) : proj_i(G^I) \rightarrow proj_i(G'^I)$ the graph morphism that only consider the $i$-labels of the  $I$-labelled graphs. 

From an $I$-indexed family of  partially labelled graphs $G_i$  defined on a common base $(V,E,s,t,\bot,l_E)$ with $l_{V,i}$ as node labelling function, we define by $Prod_{i \in I} G_i$ the  $I$-labelled graph $(V, E, s, t, (l_{V,i})_{i \in I}, l_E)$. Similarly, from an $I$-indexed family of  graph morphisms $g_i : G_i \rightarrow G'_i$ sharing the same base, we can define an $I$-labelled morphism $Prod_{i \in I} g_i$, from $Prod_{i \in I} G_i$ to $Prod_{i \in I} G'_i$, that coincides with any $g_i$ on the node set $V_G$ and the arc set $E_G$. Obviously, we then get the identities: $G^I = Prod_{i \in I} proj_i(G^I)$ for $G$ a  $I$-labelled graph and $g^I = Prod_{i \in I} proj_i(g^I)$ for $g^I$ an $I$-labelled morphism.

Since from any  partially labelled graphs $F$, $G$, $H$, ... and from any morphisms on them $f : F \rightarrow G$, $g : G \rightarrow H$, $h : F \rightarrow H$, we can derive their corresponding base form, respectively $F_{\bot}$, $G_{\bot}$, $H_{\bot}$, ...  $f_{\bot} : F_{\bot} \rightarrow G_{\bot}$, $g_{\bot} : G_{\bot} \rightarrow H_{\bot}$, $h_{\bot} : F_{\bot} \rightarrow H_{\bot}$, and then for any diagram made of morphisms expressed on partially labelled graphs, we can derive a similar diagram on their corresponding base. For example, from the diagram $ F \stackrel{f}{\rightarrow} G \stackrel{g}{\rightarrow} H =  F \stackrel{h}{\rightarrow} H$, we can derive the diagram $ F_{\bot} \stackrel{f_{\bot}}{\rightarrow} G_{\bot} \stackrel{g_{\bot}}{\rightarrow} H_{\bot} =  F_{\bot} \stackrel{h_{\bot}}{\rightarrow} H_{\bot}$.

\begin{lemma} 
\label{samebase}
For $m=1$, $2$, let us consider $f_m : A_m \rightarrow B_m$ and $g_m : A_m \rightarrow C_m$ two graph morphisms in ${\cal G}$ such that $g_m$ is injective and for all $x$ in $V_{B_m}$ (resp. in $E_{B_m}$), 
$\{l_{B_m, V}(x)\} \cup l_{C_m,V}({g_V}_m({f_V}_m^{-1}(x)))$  (resp. $\{l_{B_m,E}(x)\} \cup l_{C_m,E}({g_E}_m({f_E}_m^{-1}(x)))$) contains at most one element, then there exists a graph $D_m$ and graph morphisms $f'_m : C_m \rightarrow D_m$ and $g'_m : B_m \rightarrow D_m$ such that the following diagram is a pushout\footnote{A commutative diagram $A \stackrel{g}{\rightarrow} C \stackrel{f'}{\rightarrow} D = A \stackrel{f}{\rightarrow} B \stackrel{g'}{\rightarrow} D$ is a  pushout if and if for every graph $X$ and all morphisms $h : B \rightarrow X$ and $k : C \rightarrow X$ with $k \circ g = h \circ f$, there is an unique morphism $x : D \rightarrow X$ with $x \circ g' = h$ and $x \circ f' = k$.}
   \vspace{-0.1cm}
\begin{center}~
\xymatrix{  A_m \ar[d]^{g_m} \ar[r]_{f_m} & B_m \ar[d]_{g'_m}	\ar@/^/[ddr]^{h_m} & \\
C_m \ar[r]^{f'_m}  \ar@/_/[drr]_{k_m} &D_m \ar[dr]^{x_m} & \\
& & X_m \\	}
\end{center}

Moreover, if both pushout diagrams  have the same underlying base diagram, that is 
${A_1}_{\bot} = {A_2}_{\bot}$, ${B_1}_{\bot} = {B_2}_{\bot}$, ${C_1}_{\bot} = {C_2}_{\bot}$, 
${f_1}_{\bot} = {f_2}_{\bot}$ and ${g_1}_{\bot} = {g_2}_{\bot}$, then we get ${D_1}_{\bot} = {D_2}_{\bot}$, ${f'_1}_{\bot}  = {f'_2}_{\bot}$  and ${g'_1}_{\bot} = {g'_2}_{\bot}$.  
\end{lemma}

\begin{proof}
The proof of the existence of pushout is given in \cite{Habel-Plump02}. 

The uniqueness of the base elements   ${D_m}_{\bot}$, ${f'_m}_{\bot}$  and ${g'_m}_{\bot}$ comes from the fact that the proof in \cite{Habel-Plump02} explicitly constructs the elements ${D_m}_{\bot}$,  ${f'_m}_{\bot}$  and ${g'_m}_{\bot}$ in relation to the elements of the base diagram. 
\end{proof}

For convenience issues, we note $B_m +_{A_m} C_m$ the graph $D_m$, occurring in the pushout diagram.   

\begin{lemma}[Existence of pushouts]
Let $f^I : A^I \rightarrow B^I$ and $g^I : A^I \rightarrow C^I$ be two $I$-labelled morphisms in ${\cal G}^I$ such that $g^I$ is injective and for all $x$ in $V_{B}$ (resp. in $E_{B}$), for all $i$ in $I$, 
$\{{l_{B,V,i}}(x)\} \cup {l_{C,V,i}}({g_{V,i}}({f_{V,i}}^{-1}(x)))$  (resp. $\{{l_{B,E}}(x)\} \cup {l_{C,E}}({g_E}({f_E}^{-1}(x)))$) contains at most one element, then there exists a  $I$-labelled graph $D^I$ and two $I$-labelled morphisms $f'^I: C^I \rightarrow D^I$ and $g'^I : B^I \rightarrow D^I$  in ${\cal G}^I$ such that the following diagram is a pushout:
   \vspace{-0.6cm}
\begin{center}~
\xymatrix{  A^I \ar[d]^{g^I } \ar[r]_{f^I} & B^I \ar[d]_{g'^I}	\ar@/^/[ddr]^{h^I} & \\
C^I \ar[r]^{f'^I}  \ar@/_/[drr]_{k^I} &D^I \ar[dr]^{x^I} & \\
& & X^I \\	}
\end{center}
Moreover $D^I$ can be defined as $Prod_{i \in I}  D_i$ with $D_i = proj_i(B^I) +_{proj_i(A^I)} proj_i(C^I)$
\end{lemma}
 
 \begin{proof}
 By lemma~\ref{samebase}, we know that $(D_i)_{i \in I}$ have the same base because $(proj_i(A))_{i \in I}$,  $(proj_i(B))_{i \in I}$ and  $(proj_i(C))_{i \in I}$ have respectively the same base. Thus, $Prod_{i \in I}  D_i$ is a well defined  $I$-labelled graph.
 
Moreover, there exist $I$-labelled  morphisms $f'^I : C^I \rightarrow D^I$ and $g'^I : B^I \rightarrow D^I$ in ${\cal G}^I$ ensuring that the diagram is commutative. It suffices to choose : $f'^I =  Prod_{i \in I} f'_i$ and $g'^I =  Prod_{i \in I} f'_i$ where $f'_i:  proj_i(B) \rightarrow D_i$ and $g'_i : proj_i(C) \rightarrow D_i$ are the underlying morphisms constituting the pushout construction :  $D_i = proj_i(B) +_{proj_i(A)} proj_i(C)$.

Let us show the universal property : let us consider  $k^I : C^I \rightarrow X^I$ and $h^I : B^I \rightarrow X^I$ two  $I$-labelled graphs with $h^I \circ f^I = k^I \circ g^I$. By the universal property of $D_i$, there exists a unique  labelled morphism $x_i :  D_i \rightarrow proj_i(X^I)$ such that $x_i \circ proj_i(f'^I) = x_i \circ proj_i(g'^I)$. Then we can consider $x^I = Prod_{i \in I} x_i : D^I \rightarrow X^I$ verifying $x^I \circ f'^I = x^I \circ g'^I$.  
 \end{proof}
 
 Thus, constructions holding on partially labelled graphs can be replicated at the level of  $I$-labelled graphs. It suffices to work with their $i$-components, index per index, using the $proj_i$ application and to reconstruct $I$-labelled graphs or morphisms by applying the $Prod_{i \in I}$ operator on objects sharing the same base.
 
In the sequel, we take benefit of all results given in \cite{Habel-Plump02} : existence of pullbacks, characterisation of natural pushouts\footnote{A natural pushout is both a pushout and a pullback.}. For the purpose of simplicity, we give up the exponent $I$ upon the  $I$-labelled graph (resp. morphism) names and we will use $I$-labelled inclusions to define rules.

 \begin{definition}[graph transformation rule]
\label{def:regle}
A {\em graph transformation rule} $r : L \hookleftarrow K \hookrightarrow R$ over ${\cal G}^I$ consists of two $I$-labelled graph inclusions $K \hookrightarrow L$ and $K \hookrightarrow R$  in ${\cal G}^I$ such that:
\begin{enumerate}
\item
for all node $x \in V_L$ and all $i \in I$, $l_{L,V,i}(x) = \bot$ implies $x \in V_K$ and $l_{R,V,i}(x) = \bot$;
reciprocally, for all node $x \in V_R$ and all $i \in I$, $l_{R,V,i}(x) = \bot$ implies $x \in V_K$ and $l_{L,V,i}(x) = \bot$;
\item
for all arc $x \in E_L$, $l_{L,E}(x) = \bot$ implies $x \in E_K$ and $l_{R,E}(x) = \bot$;
reciprocally, for all arc $x \in E_R$, $l_{R,E}(x) = \bot$ implies $x \in E_K$ and $l_{L,E}(x) = \bot$;
\end{enumerate}
Usually, $L$ is called the {\em left-hand side}, $R$ the {\em right-hand side} and $K$ the {\em kernel}.
\end{definition}

\begin{definition}[direct transformation]
\label{def:transformation_directe}
Let $r : L \hookleftarrow K \hookrightarrow R$ be a graph transformation rule over ${\cal G}^I$ and $G$ a $I$-labelled graph and $m : L \rightarrow G$ an injective $I$-labelled morphism in ${\cal G}^I$ called {\em match morphism}.

A {\em direct transformation} $G \overset{r,m}{\Rightarrow} H$ of $G$ into $H$ consists in the following natural double pushout defined over ${\cal G}^I$~:
   \vspace{-0.5cm}
\begin{center}~
\xymatrix{
     L \ar@{<-^{)}}[r]  \ar@{->}[d]_{m}^*+{~~(1)} &
    	 K \ar@{^{(}->}[r]  \ar@{->}[d]^*+{~~(2)} &
	    R  \ar@{->}[d] \\
	G \ar@{<-^{)}}[r]   &
    	 D \ar@{^{(}->}[r]  &
	    H \\}
    \end{center}
\end{definition}

\begin{definition}[dangling condition]
An $I$-labelled morphism  $m : L \rightarrow G$ satisfies the {\em dangling condition} with respect to the  inclusion $K \hookrightarrow L$, if none node of $m(L) \backslash m(K)$ is source or target of an arc of $G \backslash m(L)$.
\end{definition}

\begin{theorem}[Existence and uniqueness of direct transformation]
\label{theo:existence_transformation}
Let $r : L \hookleftarrow K \hookrightarrow R$ be a rule and $m : L \rightarrow G$ a match morphism in ${\cal G}^I$, the previous direct transformation $G\Rightarrow^{r,m}H$ exists if and only if $m$ satisfies the dangling condition. Moreover, in this case $D$ and $H$ are unique up to isomorphism. 
\end{theorem}

As our framework of $I$-labelled graphs is a direct adaptation of  partially labelled graphs as defined in \cite{Habel-Plump02}, this theorem is directly obtained by the application to  $I$-labelled graphs and $I$-labelled morphisms of the similar theorem of  \cite{Habel-Plump02} that consider  partially labelled graphs and graph morphisms. Finally, we also inherited from \cite{Habel-Plump02} that for a derivation $G \overset{r,m}{\Rightarrow} H$, $H$ is totally labelled if and only if $G$ is totally labelled where a $I$-labelled graph is said to be totally labelled when each labelling function $l_{V,i}$ is totally labelled. To sum up, graph transformations defined over ${\cal G}$ can be easily adapted for  ${\cal G}^I$ (thus for $I$-labelled graphs) by preserving all constructions and results.

\section{G-maps}
\label{sec:2-gmap}

In this section, we introduce the definition of our embedded topological structures as a particular class of  $I$-labelled graphs. First, we consider graphs without node labels  to represent the topological structure. Then, we define node labelling functions to represent the embedding. Thus, the topological structure is encoded as the base of the  $I$-labelled graph representing the embedded topological structure.

   \vspace{-0.3cm}

\subsection{The topological graph}

\begin{wrapfigure}{r}{50mm}
    \centering
            \vspace{-1.6cm}
        \includegraphics[width=45mm]{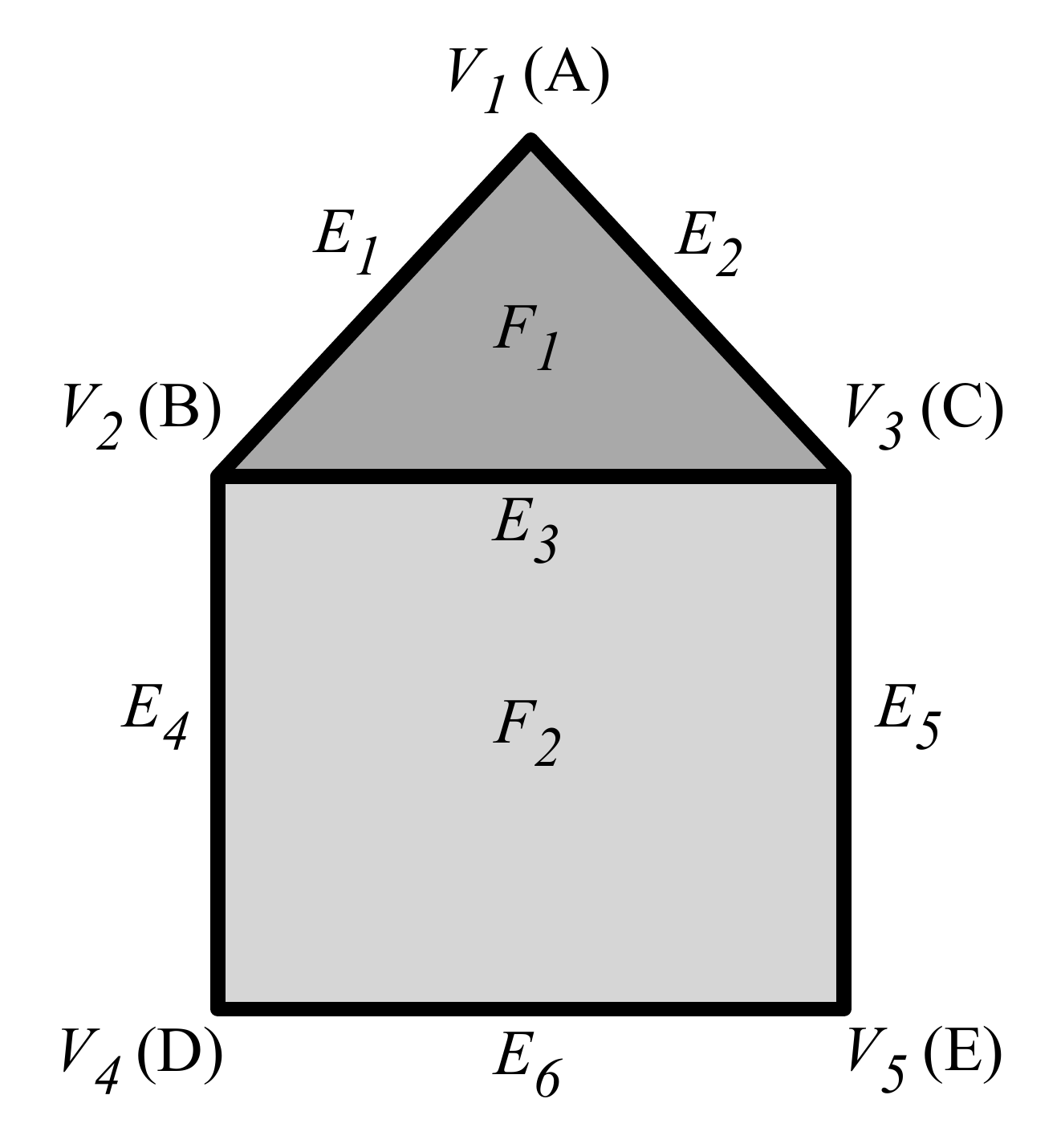}
         \vspace{-0.3cm}
    \caption{Embedded $2$D object}
    \label{fig:house_1}
\end{wrapfigure}

As said in the introduction, we choose the topological model of generalized maps (or \mbox{G-maps})~\cite{Lienhardt91}. This model is mathematically well defined. Its first main advantage is the homogeneity in the handling of dimensions: objects of any dimension can be represented in the same manner as graphs. This allows us to use rules for denoting operations defined on embedded G-maps, in an uniform way \cite{PCLAM2007_2064,PACL2008_2394}. The second advantage is that the G-map model comes with consistency constraints. They express conditions to define a topologically consistent object. Obviously, these constraints have to be maintained when operations are applied.

\begin{figure}[b]
    \centering
          \vspace{-0.9cm}
  \subfigure[]{\label{fig:house_2}
     \includegraphics[height=45mm]{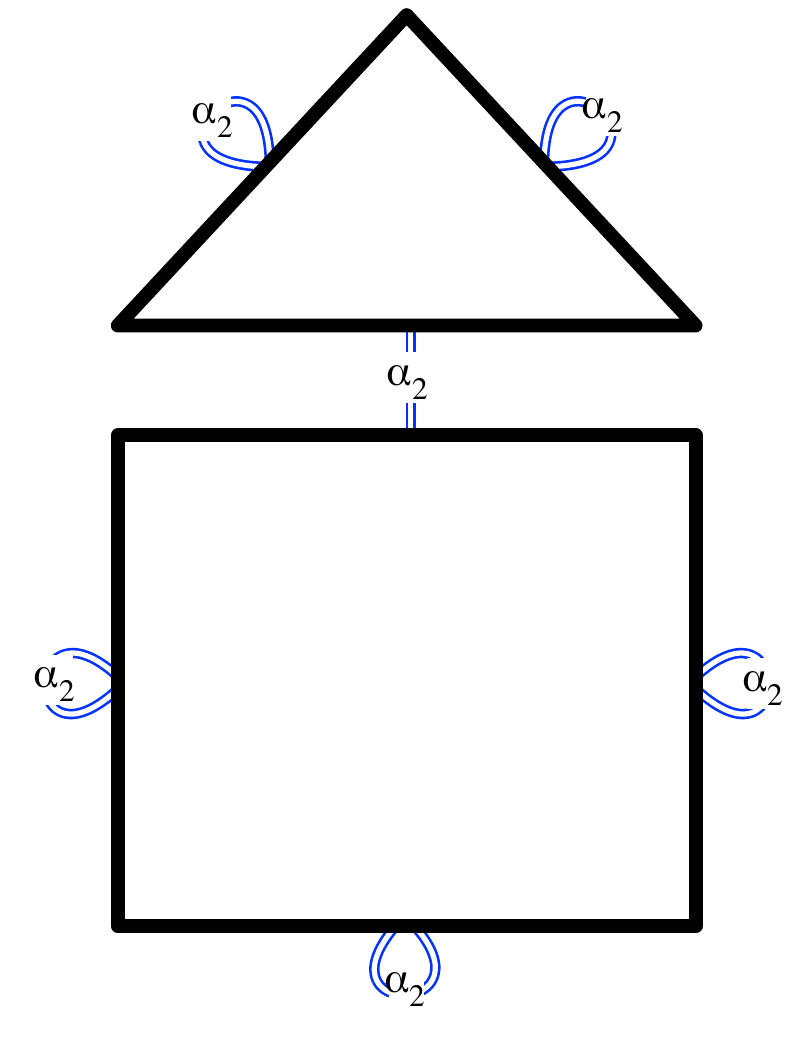}}
   \subfigure[]{\label{fig:house_3}
	\includegraphics[height=50mm]{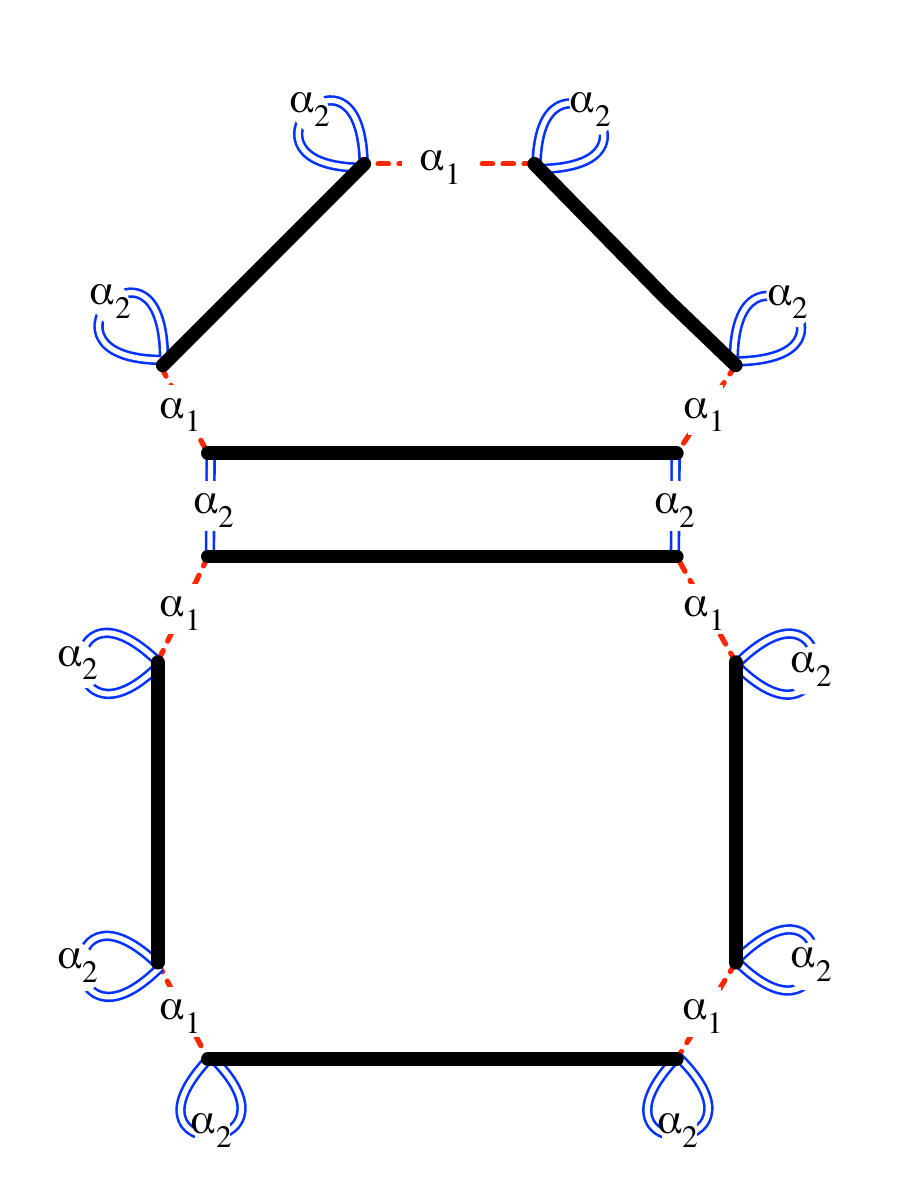}}
  \subfigure[]{\label{fig:house_4}
     \includegraphics[height=55mm]{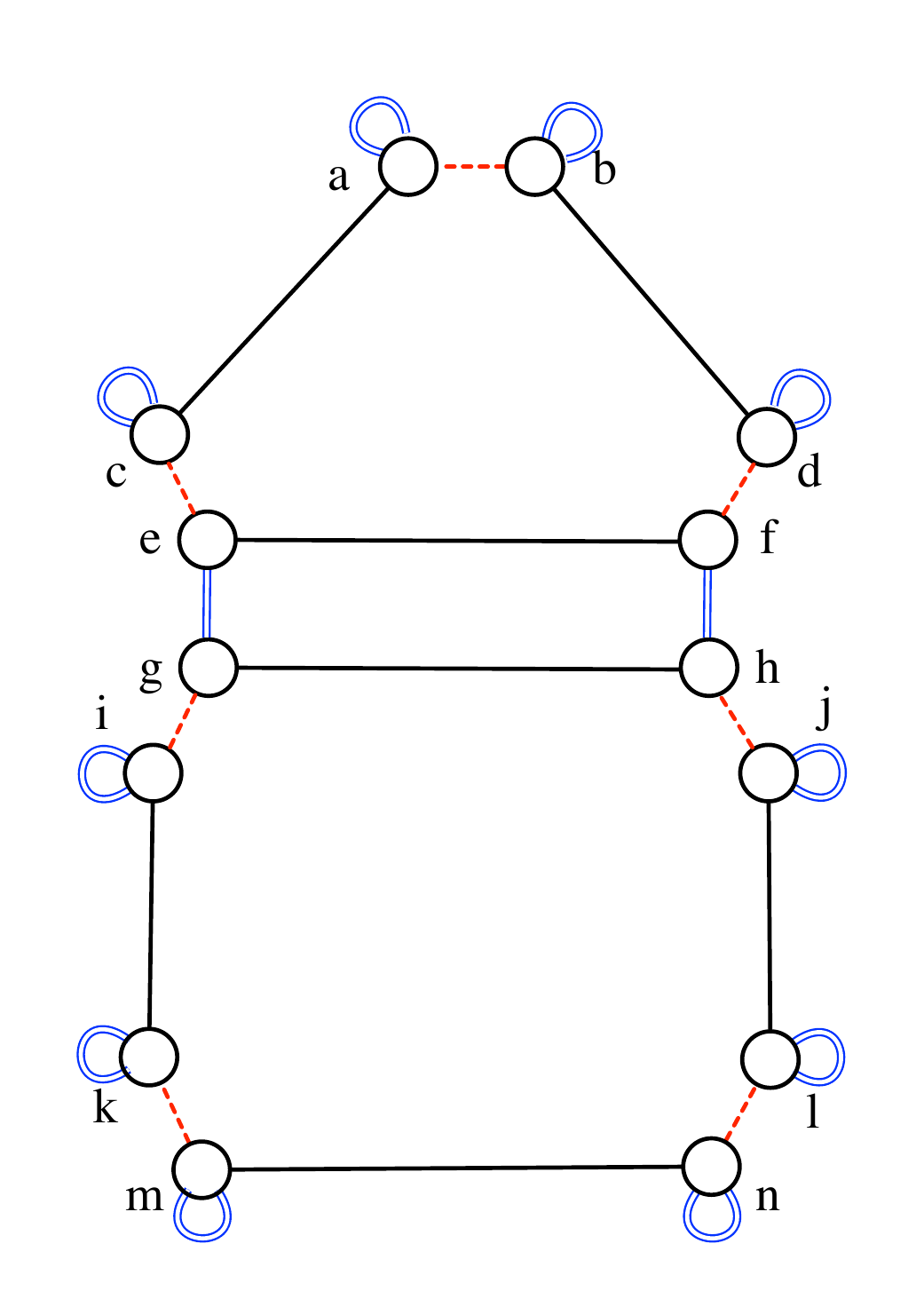}}
      \vspace{-0.3cm}
  \caption{Cell decomposition of an object}
       \label{fig:split}
        \vspace{-0.2cm}
 \end{figure}

The representation of an object as a G-map comes intuitively from its decomposition into topological cells (vertices, edges, faces, volumes, etc.). For example, the decomposition of the 2D topological object of Fig.~\ref{fig:house_1} into a $2$-dimensional G-map is shown on  Fig.~\ref{fig:split}. The object is first decomposed into faces on Fig.~\ref{fig:house_2}. These faces are \emph{linked} along their common edge $E_3$ with the relation $\alpha_2$. In the same way, faces are split into edges connected with the relation $\alpha_1$ on Fig.~\ref{fig:house_3}. At last, the edges are split into vertices by relation $\alpha_0$ to obtain the $2$-G-map of Fig.~\ref{fig:house_4}. Split vertices obtained at the end of the process are the nodes of the G-map graph and the $\alpha_i$ relations are the arcs (For a 2-dimensional G-map, $i$ belongs to $\{0,1,2\}$). Hence, for $n$ a dimension, $n$-G-maps are particular  $I$-labelled graphs where the arc label set is $\mathcal{C}_E=\{\alpha_0,\dots,\alpha_n\}$ and where arcs are totally labelled. In fact,  G-maps are represented by non-oriented graphs, that is, such that for each arc of source $v$, of target $v'$ and labelled by $\alpha_i$, there also exists an arc of source $v'$, of target $v'$ and labelled by $\alpha_i$. As usual, double reversed arcs  are represented on pictures by a non oriented arc. Notice that in all figures given in the sequel, we will use the $\alpha_i$ graphical codes of Fig.~\ref{fig:house_4} (simple line for $\alpha_0$, dashed line for $\alpha_1$ and double line for $\alpha_2$) in order to be more readable.

 \begin{wrapfigure}{r}{70mm}
    \centering
        \vspace{-0.6cm}
        \includegraphics[width=63mm]{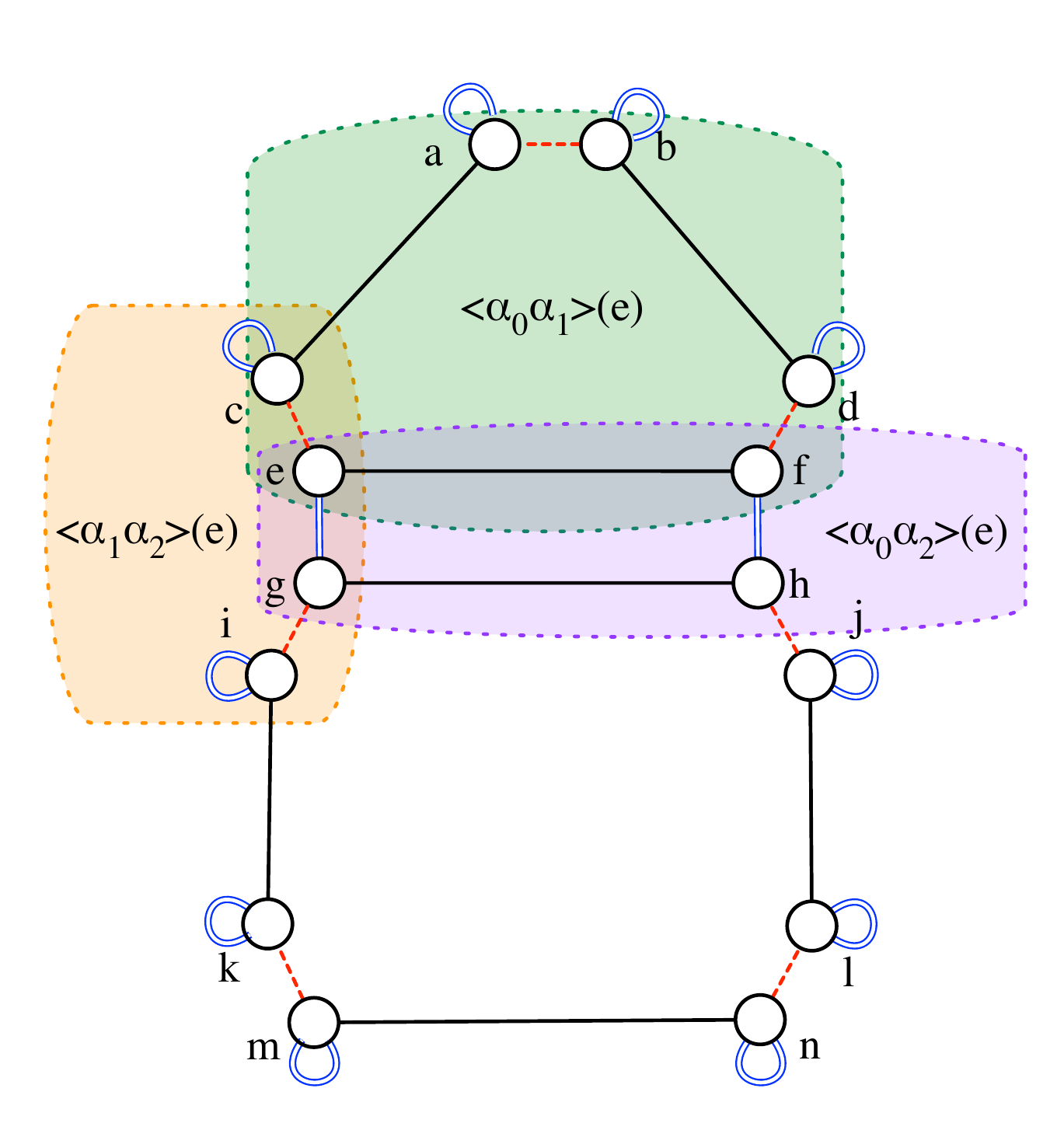}
         \vspace{-0.4cm}
    \caption{Reconstruction of adjacent cells of~$e$}
    \label{fig:cells}
\end{wrapfigure}

Topological cells are not explicitly represented in G-maps but only implicitly defined as subgraphs. They can be computed using traversal of nodes using a given set of neighborhood arcs. For example, on Fig.~\ref{fig:cells}, the $e$ incident 0-cell (or object vertex) is the subgraph which contains $e$, nodes reachable from $e$ using arcs $\alpha_1$ and $\alpha_2$ labelled (nodes $c$, $e$, $g$ and $i$) and the arcs themselves. This subgraph is denoted by ${<\!\!\alpha_1 \alpha_2\!\!>(e)}$ and models the vertex $V_2$ of Fig.~\ref{fig:house_1}. 
On Fig.~\ref{fig:cells}, the $e$ incident $1$-cell (or object edge) is the subgraph ${<\!\!\alpha_0 \alpha_2\!\!>(e)}$ containing nodes $e, f, g$ and $h$, and adjacent $\alpha_0$ and $\alpha_2$ arcs.  It represents the topological edge $E_3$. 
Finally, the $e$ incident $2$-cell (or object face) is the subgraph ${<\!\!\alpha_0 \alpha_1\!\!>(e)}$ and represents the face $F_1$.
More generally, the notion of orbit may be defined.

\begin{definition}[$n$-topological graph and orbit]
A  $I$-labelled graph $G$ is said to be an {\em $n$-topological graph}  if all arcs are  labelled in $\mathcal{C}_E=\{\alpha_0,\dots,\alpha_n\}$. 

Let us consider $o$ a subword\footnote{ $\alpha_{i_1} \dots \alpha_{i_k}$is a subword of $\alpha_0 \alpha_1 \dots \alpha_n$ if  $i_1 \dots i_k$ is a restricted increasing sequence of $[0,n]$.} of $\alpha_0 \alpha_1 \dots \alpha_n$.

Let $\equiv_{G<\!\!o\!\!>}$ be the {\em equivalence orbit relation} between $G$ nodes defined as the reflexive, symmetric and transitive closure built from arcs labelled by a label in $o$, i.e.,  ensuring that for each arc $e$ of $G$ labelled in $o$, we have $s(e) \equiv_{G<\!\!o\!\!>} t(e)$.

For any node $v$ of $G$, the $<\!\!o\!\!>$-{\em orbit}  (also simply called  {\em orbit}) of $G$ adjacent to $v$ is  denoted by ${G<\!\!o\!\!>(v)}$  and is defined as the subgraph of $G$ whose set of nodes is the equivalence class of $v$ using $\equiv_{G<\!\!o\!\!>}$, whose set of arcs are those labelled on $o$ between previous nodes, and such that source, target, labelling functions are the restrictions of the corresponding functions on sets of nodes and arcs of the equivalence class.
\end{definition}

As G-maps are mathematically well defined, they come with consistency constraints. 

\begin{definition}[Generalised map]
\label{def:topo_constr}
An $n$-dimension generalized map, or $n$-G-map, is a $n$-topological graph $G$, that satisfies the following topological constraints~:
\begin{itemize}
\item {\bf Non-orientation constraint}: $G$ is non-oriented, i.e. for each arc $e$ of $G$,  there exists a reversed arc $e'$ of $G$, such as $s_G(e')=t_G(e), t_G(e')=s_G(e)$, and $l_{G,E}(e')=l_{G,E}(e)$ ;
\item {\bf Adjacent arc constraint}: each node is the source node of exactly $n+1$ arcs respectively labelled by $\alpha_0$ to $\alpha_n$;
\item {\bf Cycle constraint}: for every $\alpha_i$ and $\alpha_j$ verifying  ${0 \leq i \leq i+2 \leq j \leq n}$, there exists a cycle\footnote{A node $v$ of a graph $G$ has an adjacent cycle labelled $l_1 \dots l_k$ if there is a path of arcs $e_1 \dots e_k$ from $v$ to $v$ such $e_1$, \dots, $e_k$ are respectively labelled by $l_1$, \dots, $l_k$.} labelled by ${\alpha_i \alpha_j \alpha_i \alpha_j}$ starting from each node.
 \end{itemize}
\end{definition}

These constraints ensure that objects represented by embedded G-maps are consistent manifolds \cite{Lienhardt94}. 
In particular, the cycle constraint ensures that in \mbox{G-maps}, two \mbox{$i$-cells}  can only be adjacent along \mbox{(${i-1}$)-cells}. For instance, in the \mbox{$2$-G-map} of Fig.~\ref{fig:house_4}, the   ${\alpha_0 \alpha_2 \alpha_0 \alpha_2}$ cycle implies that faces are stuck along topological edges. Let us notice that thanks to loops (see $\alpha_2$-loops in Fig.~\ref{fig:house_4}), these three constraints also hold at the border of objects.

 \subsection{Embedded generalized maps}

We started to define $n$-G-map as $I$-labelled graphs where the arc label set is $\mathcal{C}_E=\{\alpha_0,\dots,\alpha_n\}$.
We now complete this definition with a family of node label sets  
to represent the embedding.
Actually, as sketched in the introduction, each kind of embedding label has its own type and is defined on a particular kind of topological cell: for example, a point can be attached to a vertex, a color to a face. Thus, a node labelling function $l_{V,i}$ composing the embedding will be equipped with two static pieces of information: the kind of topological cells that is concerned by $l_{V,i}$
and the type of the data that are described by $l_{V,i}$. Based on algebraic specifications, a node labelling function is characterized by an {\em embedding operation} $\pi : <\!\!o\!\!> \rightarrow s$ where $\pi$ is its operation name, $s \in S$ is its type with $S$ a given set of data types and $<\!\!o\!\!>$ is its domain given as an $n$-dimensional orbit type.  
Hence, for a G-map, the family of node label sets  $(\mathcal{C}_{V,\pi})_{\pi \in \Pi}$ is defined by a set $\Pi$ of embedding operations. For example, for the object of Fig.~\ref{fig:house_1}, the set of embedding operations can be $\Pi=\{ point : <\!\!\alpha_1 \alpha_2\!\!> \rightarrow point\_type , color : <\!\!\alpha_0 \alpha_1\!\!> \rightarrow color\_type \}$ where $point\_type$ and $color\_type$ are supposed to be appropriate data types.  In particular, for an embedding operation $\pi : <\!\!o\!\!> \rightarrow s$,  $\mathcal{C}_{V,\pi}$ will be a set of values of type $s$, according to some algebra interpreting all the sorts involved by the embedding.

 \begin{wrapfigure}{r}{50mm}
    \centering
             \vspace{-0.8cm}
        \includegraphics[width=45mm]{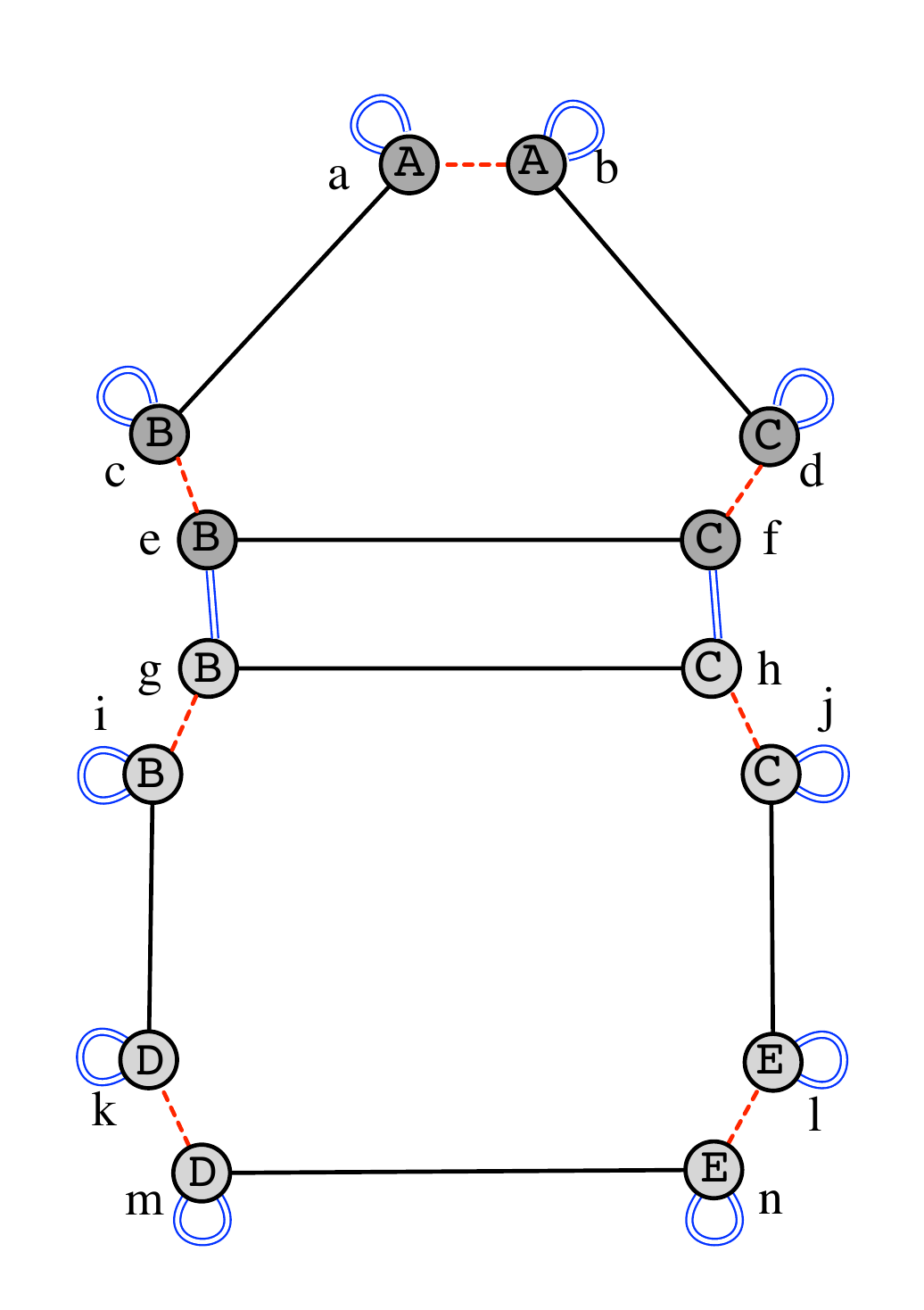}
         \vspace{-0.3cm}
    \caption{Embedded $2$-G-map}
    \label{fig:house_5}
\end{wrapfigure}

Moreover, as an embedding operation $\pi : <\!\!o\!\!> \rightarrow s$ is characterized by its domain cell, it is expected that on an embedded G-map, the $\pi$-label, also called $\pi$-embedding (that is, the image by $l_{V,\pi}$ ) is the same for every node belonging to  a common $<\!\!o\!\!>$-orbit. Hence, we represent on Fig.~\ref{fig:house_5} the embedded version of the object of Fig.~\ref{fig:house_1}. Let us notice that this graphical representation is a simplification of the full notation. For example, we only label $a$ with its point label $A$ and color it with its color label instead of the full labelling $(point:A, color:dark\_grey)$.
Hence, for the embedding operation $point$, $a$ and $b$ are labelled by $A$, $c,e,g$ and $i$ by $B$, $d,f,h$ and $j$ by $C$, $k$ and $m$ by $D$, $l$ and $n$ are labelled by $E$. For the embedding operation $color$, nodes $a$ to $f$ are labelled with dark grey and nodes $g$ to $n$ are labelled with clear grey. Thus, on Fig.~\ref{fig:house_5}, for a domain $<\!\!o\!\!>$, every node of a $<\!\!o\!\!>$-orbit has the same label.  We express this property by embedding constraints that embedded G-maps have to satisfy. 
 
 \begin{definition}[Embedded generalised map]
\label{def:ebd_constr}
Let $n$ be a dimension and $\Pi$ 
a set of embedding operations. 
An embedded $n$-dimentional generalised map on $\Pi$, or $\Pi$-embedded $n$-G-map, is an $n$-G-map $G$ which nodes are labelled by the family $(\mathcal{C}_{V,\pi})_{\pi \in \Pi}$,  that satisfies the following embedding constraint~:

 {\bf Embedding constraint}: for all embedding operations $(\pi : <\!\!o\!\!> \rightarrow s)$ of $\Pi$, 
all nodes of a given $<\!\!o\!\!>$-orbit  of $G$   are labelled with the same defined $\pi$-embedding {\em i.e.} for all nodes $v$ and $w$ of $G$, such that $v \equiv_{G<\!\!o\!\!>}w$ then $ l_{V,\pi}(w) \neq \bot$ and $l_{V,\pi}(v) = l_{V,\pi}(w)$.
\end{definition}

Clearly, $\Pi$-embedded $n$-G-maps are $\Pi$-labelled graphs. To handle and compute data associated to embedding operations, we define an algebra parameterised by a given $\Pi$-embedded $n$-G-map $G$. Let us first note $v.\pi$ the access to the $\pi$-label $l_{G,V,\pi}(v)$ of a node $v$ of $G$. For example, on the embedded G-map of Fig.~\ref{fig:house_5}, $a.point$ is $A$ and $a.color$ is dark grey. Thanks to the topological adjacent arcs constraint, we can also define link operations on G-map's nodes that from a given node, give access to neighboring nodes. So, for each node $v$ of $G$ and each arc label $\alpha_i$, $v.\alpha_i$ is the only node $v'$ of $G$ such that there exists an arc $e$ with  $s_G(e) = v$, $t_G(e)=v'$ and $l_{G,E}(e) = \alpha_i$. For example, on the embedded G-map of Fig.~\ref{fig:house_5}, $a.\alpha_1$ is the $b$ node, and $a.\alpha_0.point$ is $c.point$ {\em i.e.} B.

In the context of geometric modelling, it is common that operations collect all the $\pi$-embedding values that are carried by nodes of a given cell. For example, the triangulation of a face collects all the points associated to the face in order to compute the new point associated to the added center. Thus, we consider the collection of a given embedding operation $\pi$ carried by a given orbit $<\!\!o\!\!>(v)$. The notation $\pi\{<\!\!o\!\!>(v)\}$ will denote the multiset of $\pi$-labels of all nodes of $G<\!\!o\!\!>(v)$, that is, of  the $<\!\!o\!\!>$-orbit incident to node $v$ of $G$. For example, on the embedded G-map of Fig.~\ref{fig:house_5}, $point \{ <\!\!\alpha_0,\alpha_1,\alpha_2\!\!>(a) \}$ is the multiset $\{ A, B, C, D, E \}$ containing all points that correspond to $point$-labels of nodes of the $<\!\!\alpha_0,\alpha_1,\alpha_2\!\!>$-orbit adjacent  to the node $a$. Let us notice that our definition only keeps a point per $<\!\!\alpha_1,\alpha_2\!\!>$-cell that intersects the initial cell, here the orbit $<\!\!\alpha_0,\alpha_1,\alpha_2\!\!>(a)$. Thus,  even if the point $B$ occurs four times as $point$-embedding of nodes of $<\!\!\alpha_0,\alpha_1,\alpha_2\!\!>(a)$, that is for the nodes $c$, $e$, $g$ and $i$, there is an unique occurrence of the point $B$ in  $point \{ <\!\!\alpha_0,\alpha_1,\alpha_2\!\!>(a) \}$ since  $c$, $e$, $g$ and $i$ belong to the same $0$-cell. To summarize, for an embedding operation $\pi : <\!\!o'\!\!> \rightarrow s$,  the collect operation $\pi\{<\!\!o\!\!>(v)\}$ only keeps one $\pi$-embedding label per $<\!\!o'\!\!>$-orbit intersecting the $<\!\!o\!\!>$-orbit adjacent to $v$. Thus, the collected multiset contains a $\pi$-label twice if two different $<\!\!o'\!\!>$-orbits have the same 
$\pi$-label. In our example (cf. Fig.~\ref{fig:house_5}), each vertex has a different $point$-embedding and thus, each point appears only once in the resulting multiset.

\begin{definition}[Embedding expressions]
Let $\Pi$ be a set of embeddings for G-maps of dimension~$n$.

An embedding signature $\Sigma_\Pi = (S_\Pi, F_\Pi)$ is defined by:
\begin{itemize}
\item a set of embedding sorts $S_\Pi$ which contains at least, the predefined sort $Node$, the sort $s$ of each embedding $\pi : <\!\!o\!\!> \rightarrow s$ of $\Pi$  and the associated sort $Multi(s)$,
\item a set of embedding operations $F_\Pi$ such that each operation $f \in F_\Pi$ is equipped with its profile in $S_\Pi^* \times S_\Pi$ denoted $f : s_1 \times ... \times s_n \rightarrow s$. $F_\Pi$ contains at least:
\begin{itemize}
\item access operation $\_.\pi : Node \rightarrow s$ for each embedding $\pi : <\!\!o\!\!> \rightarrow s$ of $\Pi$, 
\item link operation $\_.\alpha_i : Node \rightarrow Node$ to any arc label $\alpha_i$, 
\item and collect operation $\pi \{ <o'> (\_) \} : Node \rightarrow Multi(s)$ for every embedding $\pi : <\!\!o\!\!> \rightarrow s$ of $\Pi$ and any orbit type $o'$ of dimension $n$.
\end{itemize}
\end{itemize}

Let ${\cal T}_\Pi(V)$ be the set of embedding terms built on $\Sigma_\Pi$ and a variable set $V$ of sort $Node$.

Let $G$ be a $\Pi$-embedded $n$-G-map. An embedding algebra ${\cal A}_{G}$ is defined by: 
\begin{itemize}
\item a set of values $A_s$ for each sort $s$ of $S_\Pi$, such that, $A_{Node}$ is the node set of $G$, and $A_{Multi(s)}$ is the multiset of $A_s$ values,

\item a function $f^{\cal A}: A_{s_1} \times ... \times A_{s_n} \rightarrow A_s$ for each operation $f: s_1 \times ... \times s_n \rightarrow s$ of $F_\Pi$, such that:
\begin{itemize}
\item 
$\_.\pi^{\cal A}$ is defined on each node $v$ of $G$ by its $\pi$-label $l_{G,V,\pi}(v)$, 
\item $\_.\alpha_i^{\cal A}$ is defined on each node $v$ of $G$ by the target $t_G(e)$ of the only arc $e$ of $G$ such $s_G(e)=v$ and $l_{G,E}(e)=\alpha_i$, 
\item and $\pi \{ <o'> (\_) \}^{\cal A}$ is defined on each node $v$ of $G$ by the multiset\footnote{Thanks to the embedding constraint verified by the embedded G-map $G$ and equivalence relationship properties, this collect interpretation is well defined.}
 $\{ l_{V,\pi}(w) ~|~w \in W/\equiv_{G<\!\!o\!\!>} \}$ where $W$ is the node set of $G<o'>(v)$ and $W/\equiv_{G<\!\!o\!\!>}$ the quotient set.
 \end{itemize}
 \end{itemize}

The interpretation $eval_\sigma(t)$ of terms $t$ of ${\cal T}_\Pi(V)$ using an assignment $\sigma$ of variables $V$ on $G$ nodes, is canonically defined with the interpretation functions of ${\cal A}_{G}$.
\end{definition}

We suppose that usual data types as $point\_type$ or $color\_type$ are provided  with usual operations as the addition operation $+$, \ldots. In the sequel, such operations are used without explicit definition. For example, the operation $mean$ computes the center of gravity of a multiset  of points (type $Multi(point)$).

   \vspace{-0.5cm}
\section{G-maps rules}
\label{sec:3-gmap_rule}

As G-maps are a particular class of $\Pi$-labelled graphs, we now investigate  how operations can be defined using graph transformation rules over ${\cal G}^I$ (see Section~1). For example, the transformation of Fig.~\ref{fig:simp_apply} adds a new vertex to the central edge of the previous object. To be consistent, rules on embedded G-maps need to preserve both the topological consistency and the embedding consistency. In this section, we will give some conditions on rules to ensure the preservation of constraints in relation with  topology  and embedding. In particular, this will allow us to state that the rule of Fig.~\ref{fig:simp_apply}  can be safely applied to any embedded $G$-map, since the resulting graph is also an embedded G-map by construction. These conditions will be extended in Section~\ref{sec:4-rule_scheme} to allow the user to use variables in order to handle rules that are generic with respect to the embedding values.

\begin{figure}[h]
    \centering
  \includegraphics[width=15cm]{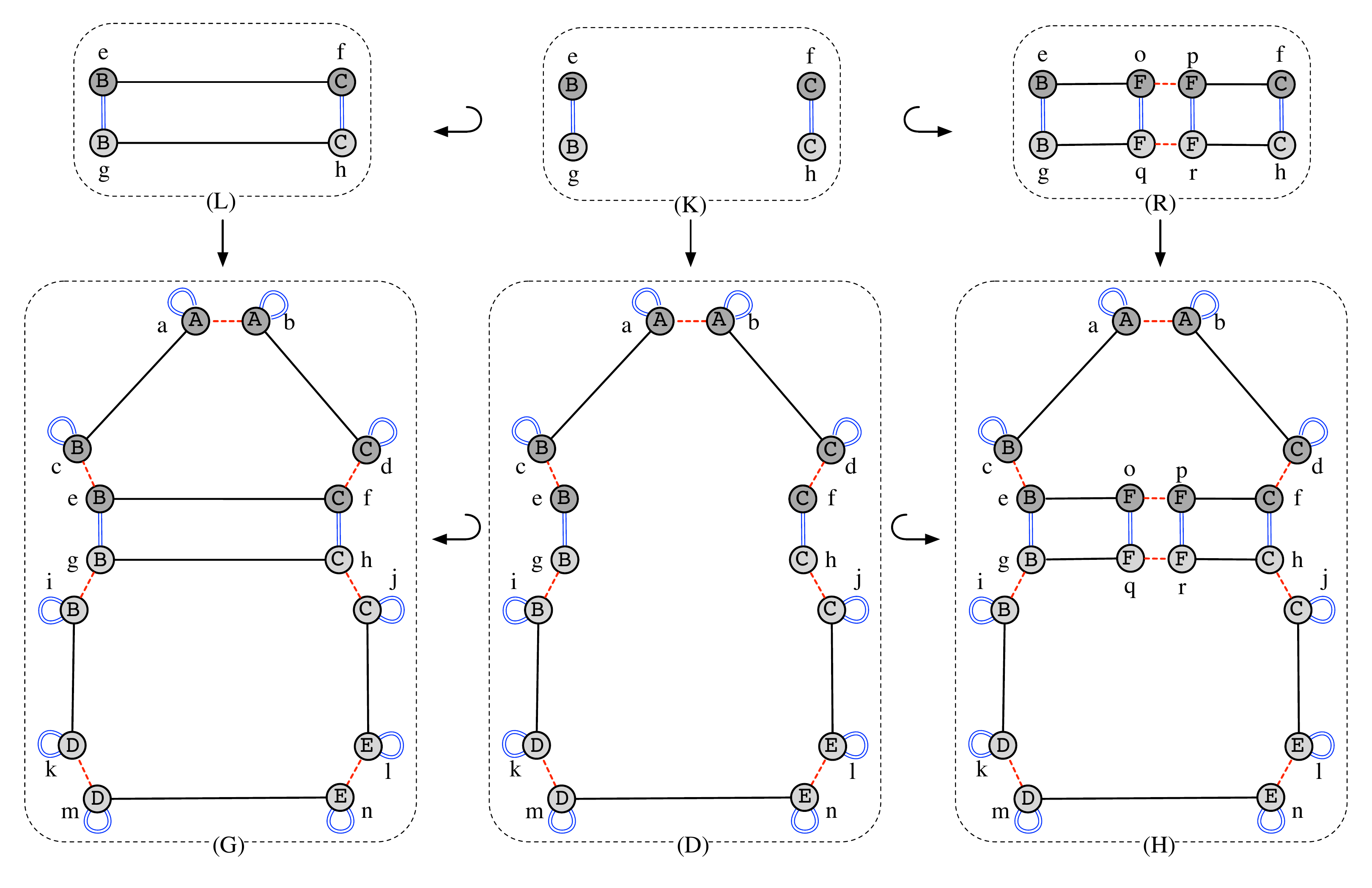}
\caption{A simple G-map transformation}
  \label{fig:simp_apply}
 \end{figure}

To ensure the topological consistency,  we have defined in~\cite{Poudret09} the following syntactic conditions on rules.
\begin{definition}[Topological consistency preservation]
\label{def:ebd_preserv}
For a rule  $r: L \hookleftarrow K \hookrightarrow R$ over ${\cal G}^{\bot}$, the conditions of topological consistency preservation are:
\begin{itemize}
\item {\it Non-orientation condition}: both $L$, $K$ and $R$ are non-oriented graphs;
\item {\it Adjacent arcs condition}: 
	\begin{itemize}
	\item adjacent arcs of preserved nodes of $K$ have the same labels on both the left-hand side and right-hand side;
	\item removed nodes of $L \backslash K$ and added nodes of $R \backslash K$ must have exactly $n+1$ adjacent arcs respectively labelled with $\alpha_0$ to $\alpha_n$;
	\end{itemize}
\item {\it Cycles condition}: 
	\begin{itemize}
	\item an added node of $R \backslash K$  must have with all ${\alpha_i \alpha_j \alpha_i \alpha_j}$-labelled cycle for ${0 \leq i \leq i+2 \leq j \leq n}$;
	\item if a preserved node of $K$ belongs to a ${\alpha_i \alpha_j \alpha_i \alpha_j}$-labelled cycle in $L$, it must belong to an ${\alpha_i \alpha_j \alpha_i \alpha_j}$-labelled cycle in $R$;
	\item if a preserved node of $K$ belongs to an incomplete ${\alpha_i \alpha_j \alpha_i \alpha_j}$-labelled cycle in $L$,  then its $\alpha_i$ and $\alpha_j$-labelled arcs are preserved in $R$.
	\end{itemize}
\end{itemize}
\end{definition}

In the following, only rules that satisfy these topological conditions are considered.  Below, we introduce syntactic conditions that ensure the embedding consistency of constructed objects.

\begin{theorem}[preservation of the embedding consistency]
\label{theo:ebd_preserv}
Let $r : L \hookleftarrow K \hookrightarrow R$ be a graph transformation rule over ${\cal G}^I$ that satisfies conditions of topological consistency preservation, $G$ a $\Pi$-embedded G-map and $m:L \rightarrow G$ a match morphism. The direct transformation $G \Rightarrow^{r,m} H$ produces an $\Pi$-embedded G-map $H$ if the following {\em conditions of embedding consistency preservation} are satisfied, for all embedding $\pi:<\!\!o\!\!> \rightarrow s \in \Pi$:
\begin{itemize}
\item All nodes of an $<\!\!o\!\!>$-orbit of $R$ are labelled with the same $\pi$-embedding, defined or not - i.e. for all nodes $v$ and $w$ of $R$ such that $v \equiv_{R<\!\!o\!\!>}w$, either $l_{R,V,\pi}(v) = l_{R,V,\pi}(w)$ with $l_{R,V,\pi}(v)  \neq \bot$, or they are both not labelled $l_{R,V,\pi}(v)  = \bot$  and $l_{R,V,\pi}(w)  = \bot$.

\item If a node $v$ of $R$ is an added node of $R \backslash K$ or a preserved node of $K$ such that its $\pi$-label is changed, then $R<\!\!o\!\!>(v)$ is a complete orbit - i.e. if $v \in V_R \backslash V_K$ or $v \in V_K$ with $l_{L,V,\pi}(v) \neq l_{R,V,\pi}(v)$, then every node of $R<\!\!o\!\!>(v)$ is the source of exactly one arc labelled by $\alpha_i$ for each label $\alpha_i$ of $o$.
\end{itemize}
\end{theorem}

\begin{figure}[h]
    \centering
  \subfigure[Incomplete redefinition]{\label{fig:simp_cond_1}
  \includegraphics[width=12cm]{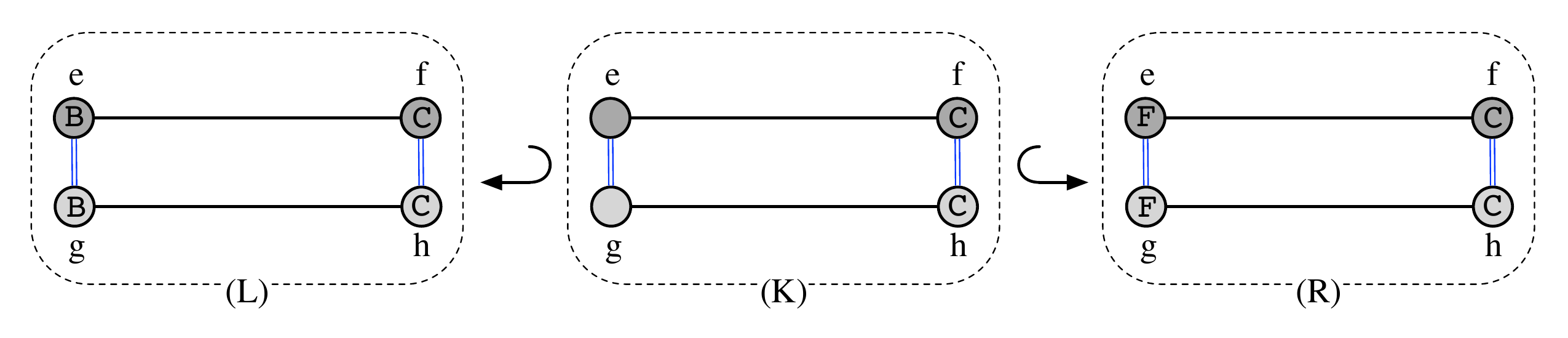}}
  \subfigure[Non-consistent added vertex]{\label{fig:simp_cond_2}
  \includegraphics[width=12cm]{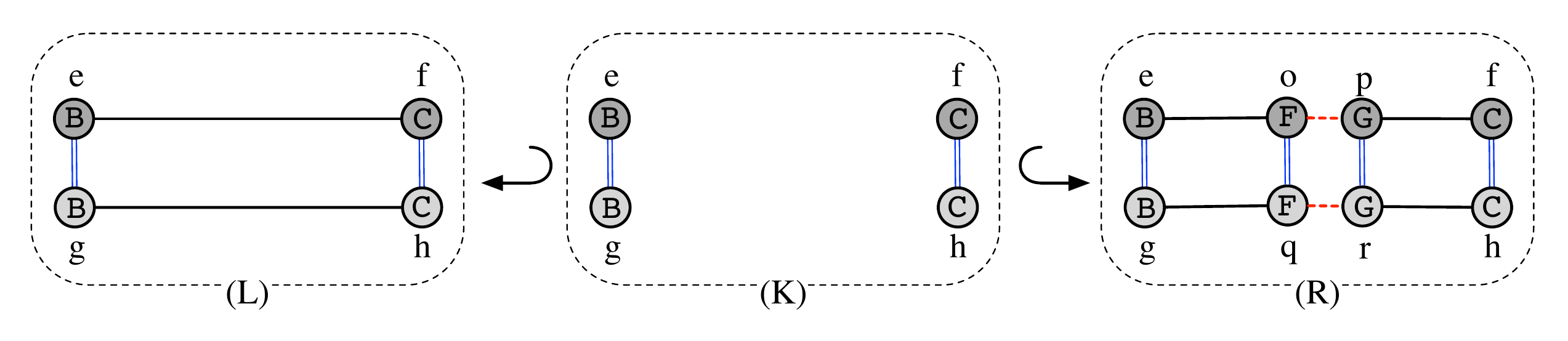}}
   \vspace{-0.3cm}
  \caption{Two non-consistent rules, not satisfying conditions of Th.~\ref{theo:ebd_preserv}.}
  \label{fig:simp_cond}
 \end{figure}

These conditions prevent the partial redefinition of an embedding. For example, the rule of Fig.~\ref{fig:simp_cond_1} tries to redefine the point $B$ by $F$. But the topological vertex (defined as a $<:\alpha_1,\alpha_2>$-orbit) is not fully matched by the rule ($\alpha_1$ is missing) and so it cannot be applied on the G-map of Fig.~\ref{fig:house_5} without breaking the embedding constraints. Indeed, if the rule was applied, node $e$ and $g$ would be labelled by point $F$ while $c$ and $i$ would still be labelled by point $B$.  In the same way, the rule of Fig.~\ref{fig:simp_cond_2} would add to the G-map a non-consistent new vertex embedded with two different points $F$ and $G$.

\begin{proof} The proof of this theorem can be found in the technical report~\cite{rapport} which contains the full length version of this paper.\end{proof}
   \vspace{-0.5cm}
\section{G-map rule schemes}
\label{sec:4-rule_scheme}

Simple rules on G-maps are quite limited. Actually, in the general context of graph transformations, rules without variables are sufficient if it is possible to write all possible transformations. In the context of geometric modeling, both the topological graph structure and the embedding node labelling are not predefined. The topological transformation depends on the original shape of the cell to transform (its number of vertices, edges, etc.). This issue has been solved by \cite{PACL2008_2394,Poudret09} with the introduction of rule schemes based on topological variables. These variables allow us to represent both the matched topological cells and their transformations. For example, a topological variable of type $<\alpha_0,\alpha_1>$  can represent any arbitrary 2-cell such that the topological triangulation operation can be applied to a triangle, a square or a pentagon.  A topological rule scheme is then instantiated according to a substitution of the given variable by a 2-cell of the G-map to be transformed. Such an instantiation builds a transformation rule that meets the conditions of topological consistency preservation (provided that the scheme rule also meets some conditions  given in \cite{PACL2008_2394,Poudret09}).
In the same way, the embedding transformation depends on the original embedding of the matched cell. For example, usually, when a face is triangulated, the central position of the added vertex depends on the positions of existing vertices. With the simple framework of Section~\ref{sec:3-gmap_rule}, there should be as many rules as possible vertex positions. We introduce embedding variables to get rule schemes that will be instantiated according to the different possible values associated to the variables.

These variables are based on the notion of {\em attributed variables} introduced by \cite{hoffmann}.
The variables label nodes of the left-hand side of rules in order to match the existing labels of the object. In the right-hand side, new labels are defined as expressions upon these variables. These algebraic expressions are then interpreted when rules are applied. For example in~Fig.~\ref{fig:attributed_var}, the variables $x$ and $y$ of the left-hand side can match any labels and the expression ${x + y}$ of the right-hand side should be evaluated according to the values provided by the match morphism in order to define the label of the new node~4. To apply this rule to an object, we instantiate the variables of the rule with the corresponding values of the matched object to obtain a classical rule that is applied as a direct transformation.

\begin{figure}[h] 
    \centering
          \vspace{-0.2cm}
     \includegraphics[height=22mm]{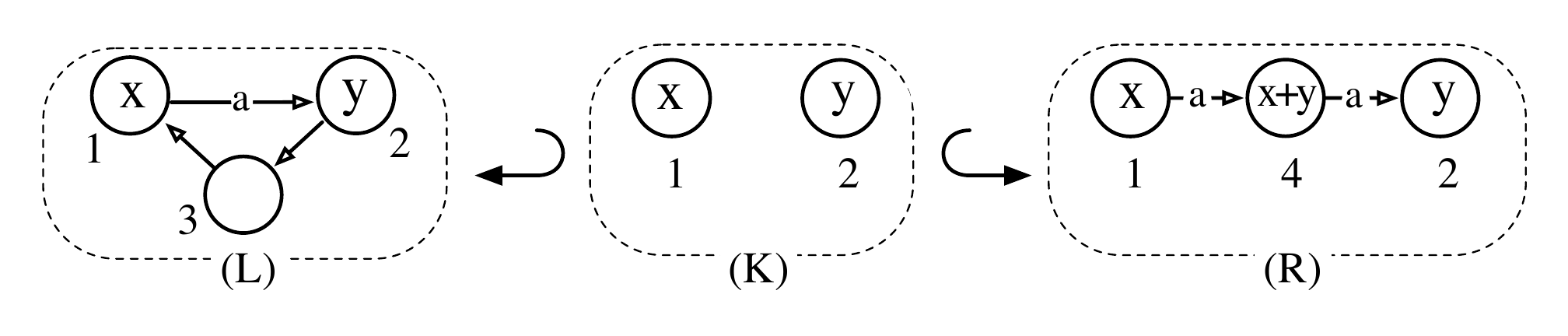}
           \vspace{-0.2cm}
   \caption{A rule with attributed variables}
   \label{fig:attributed_var}
 \end{figure}
 
As in our case, nodes have multiple labels, rules can have a variable per node and per embedding operation to define transformations. To simplify computations on embedding values, we use embedding expressions introduced in Section~\ref{sec:2-gmap}. For example, on Fig.~\ref{fig:simp_nota_expr_1}, the rule translates by a vector $\vec{P}$ the points associated to the nodes $a$ and $b$. The color associated to node $b$ is redefined while the color associated to $a$ is not matched by the rule and, as a consequence, not transformed. On Fig.~\ref{fig:simp_nota_expr_2} we use a simplified notation. As there is no ambiguity on the type of the expressions, they are not explicitly typed. In the same way, the unmatched color of $a$ is not represented. Moreover, for lack of space, the expressions will often be placed below the graph and referenced by a number. For example,  the node $a$ is labelled by the number 1 that represents the expression $a.point + \overrightarrow{P}$ associated to $(1)$.

\begin{figure}[t]
    \centering
              \vspace{-0.4cm}
  \subfigure[Full notation]{\label{fig:simp_nota_expr_1}
  \includegraphics[width=12cm]{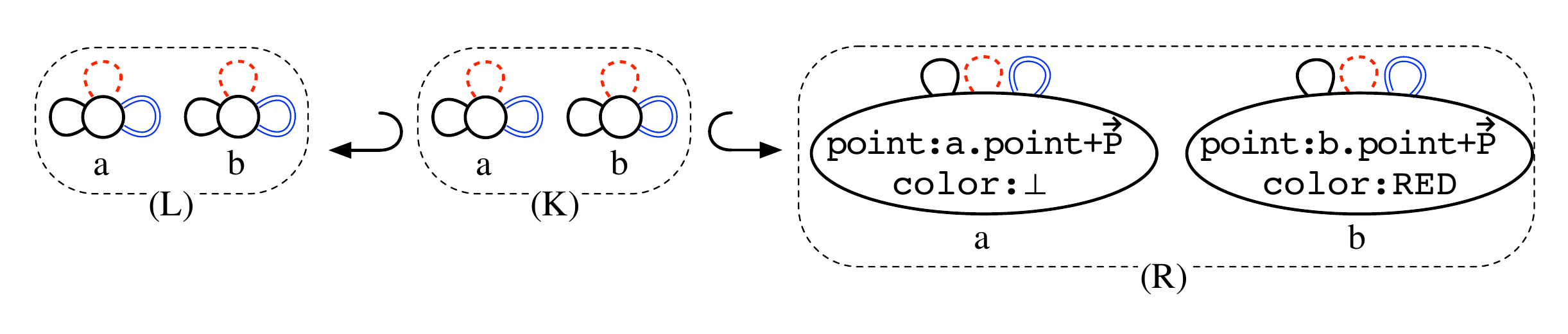}}
  \subfigure[Simplified notation]{\label{fig:simp_nota_expr_2}
         \vspace{-0.2cm}
  \includegraphics[width=10cm]{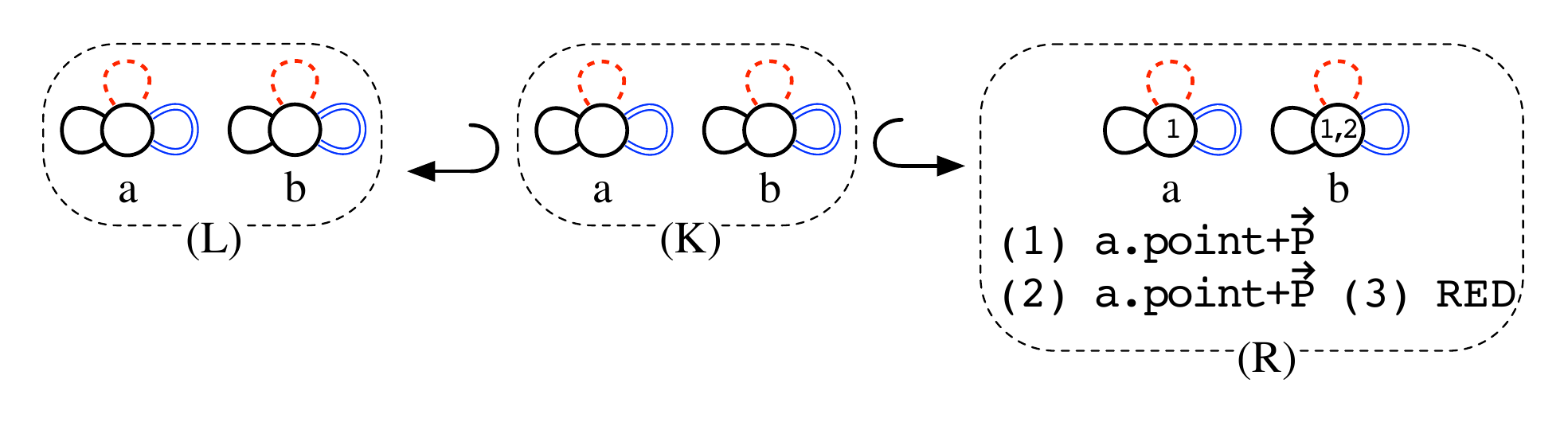}}
  \caption{Translation of an isolated vertex}
         \vspace{-0.4cm}
  \label{fig:simp_nota_expr}
 \end{figure}

Let us notice that in the example of Fig.~\ref{fig:simp_nota_expr}, this notation allows us to not explicitly label both the left-hand side  and the kernel of the rule in order to match the embedding. Expressions on variable names allow us to directly compute new labels in the right-hand side. For example, on Fig.~\ref{fig:split_scheme}, when the edge is split, the center is computed with the expression $(e.point+f.point)/2$ while the preserved nodes keep their original embedding. In order to apply the rule of Fig.~\ref{fig:split_scheme} to object of Fig.~\ref{fig:house_5} along the inclusion match morphism, the variables have to be instantiated and expressions computed. 
For example on Fig.~\ref{fig:split_scheme_inst}, $e.color$ and $g.color$ are respectively instantiated by dark grey and light grey and the new point is computed as $(B+C)/2$. However,  even with such evaluation and computation mechanisms, the rule cannot be  directly applied. The instantiation mechanism has also to complete the orbits of redefined embedding values. Indeed, rule schemes describe the modification in a  minimal way. In particular, for an embedding operation $\pi :  <\!\!o\!\!> \rightarrow s$, we have to deal with indirect modifications for nodes belonging to an $<\!\!o\!\!>$-orbit of a node whose $\pi$-embedding is modified by the rule.  For example, as the $color$-embedding labels are redefined for the node $e$, $f$, $g$ and $h$ (in the present case they remain the same), then, potentially, the $color$-embedding of all nodes that belong to an $<\alpha_0,\alpha_1>$-orbit of one of these nodes can be modified by the transformation rule application.  For this reason, for a given match morphism,  the instantiation mechanism will both substitute the embedding variables and complete the pattern under modification to include all possible indirect modifications (in Fig.~\ref{fig:split_scheme_inst},  the completion mechanism will consider the full triangle and the full square in order to redefine colors).
The application of the instantiated rule to the object is then the classical rule application (as described in Section~\ref{sec:3-gmap_rule}). 

\begin{figure}[t]
    \centering
          \vspace{-0.3cm}
  \subfigure[Rule scheme]{\label{fig:split_scheme}
  \includegraphics[width=12cm]{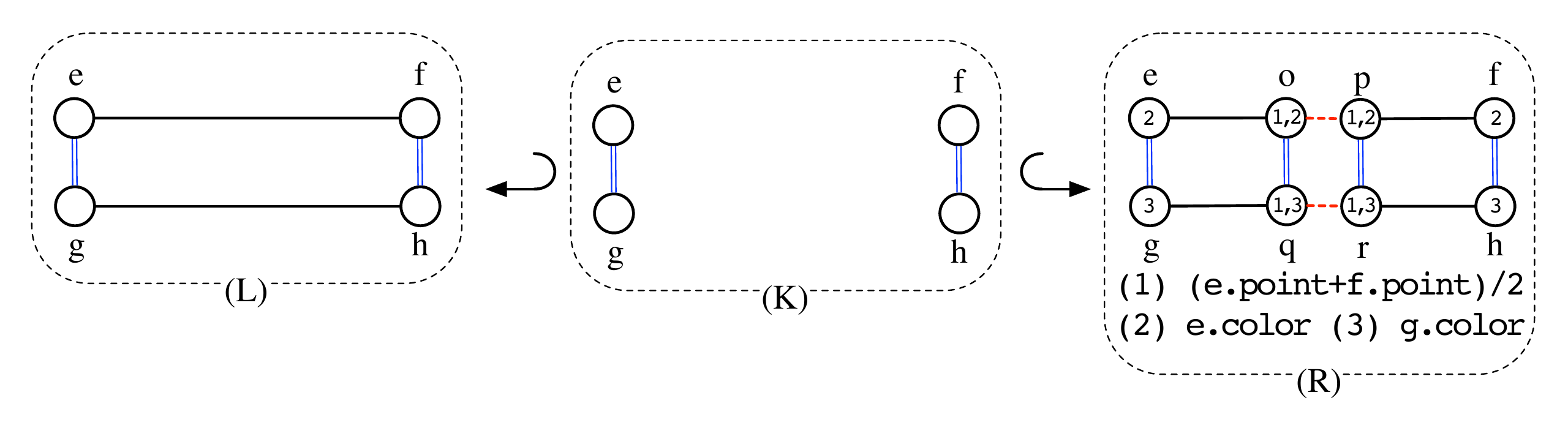}}
  \subfigure[Instantiated rule]{\label{fig:split_scheme_inst}
  \includegraphics[width=15cm]{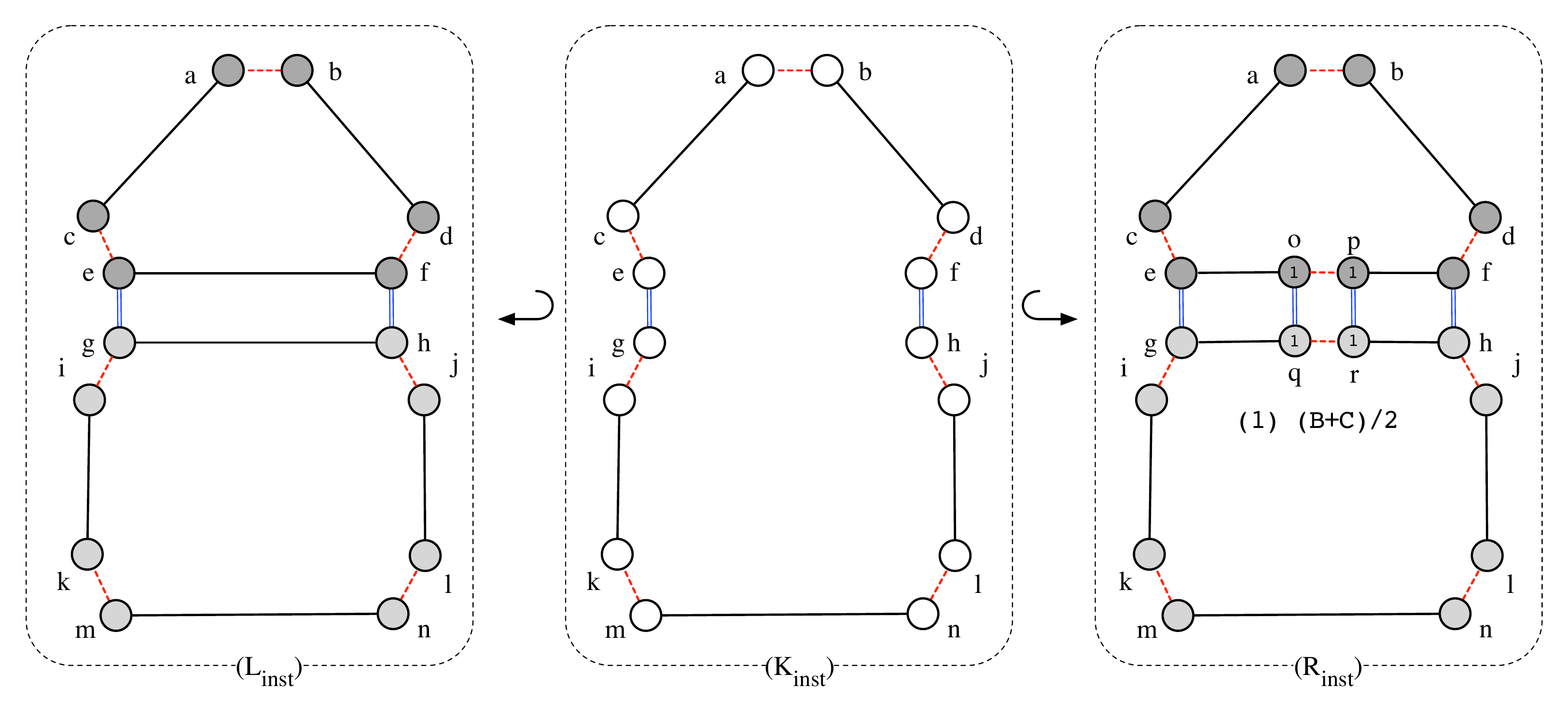}}
  \caption{Edge splitting scheme}
       \vspace{-0.3cm}
 \end{figure}
 
The rule schemes allow us to compute new embedding values by using expressions introduced in Section~\ref{sec:2-gmap}.
For example, the rule scheme of Fig.~\ref{fig:triang_scheme_1} defines the triangulation of a triangle. A vertex is added at the center of the face, and its associated point is defined by the expression  $mean(point\{<\!\!\alpha_0 \alpha_1\!\!>(a)\})$ as the mean of the points of the face. This expression is interpreted by $mean\{A,B,C\}$ when rule is instantiated on Fig.~\ref{fig:triang_scheme_inst} to be applied on object Fig.~\ref{fig:house_1}. Simultaneously, the colors of faces created by triangulation are defined as the mean between the original face color and the color of their respective adjacent faces. For example, the left/up side face color is defined as $(a.color + a.\alpha_2.color)/2$ where the expression $a.\alpha_2$ represents a node of the adjacent face (or the node itself if there is no adjacent face).
 When this rule is instantiated on Fig.~\ref{fig:triang_scheme_inst}, $a.\alpha_2.color$ is instantiated by the color of $a$, $b.\alpha_2.color$ by the color of $b$ and $e.\alpha_2.color$ by the color of $g$.  Let us notice that for this instantiation, the face is fully matched by the rule scheme and so the face orbit does not have to be completed to define the color properly.  At the opposite, the vertex orbits corresponding to the embedded points $B$ and $C$ are not fully matched but they have to be completed with $g$, $h$, $i$ and $j$ by the instantiation mechanism since $point$-embeddings are redefined for the nodes $e$ and $f$. 

\begin{figure}[t]
    \centering
              \vspace{-0.3cm}
  \subfigure[Rule scheme]{\label{fig:triang_scheme_1}
  \includegraphics[width=15cm]{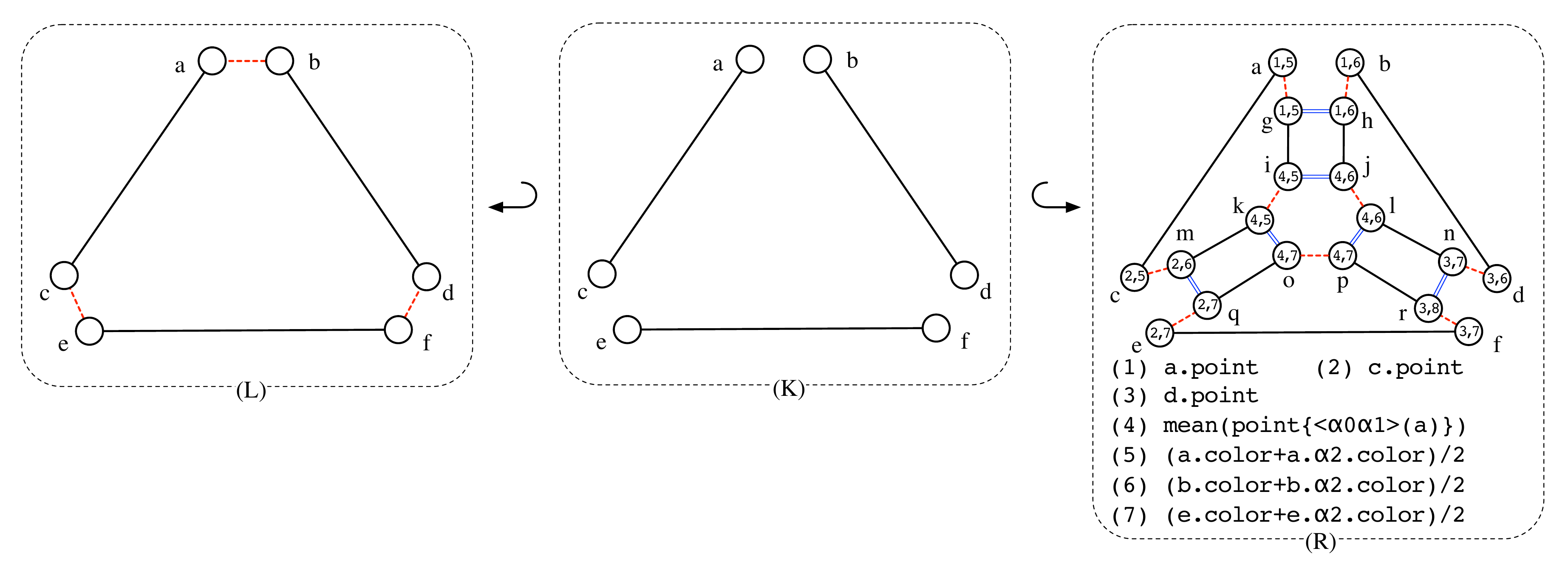} }
  \subfigure[Instantiated rule]{\label{fig:triang_scheme_inst}
  \includegraphics[width=15cm]{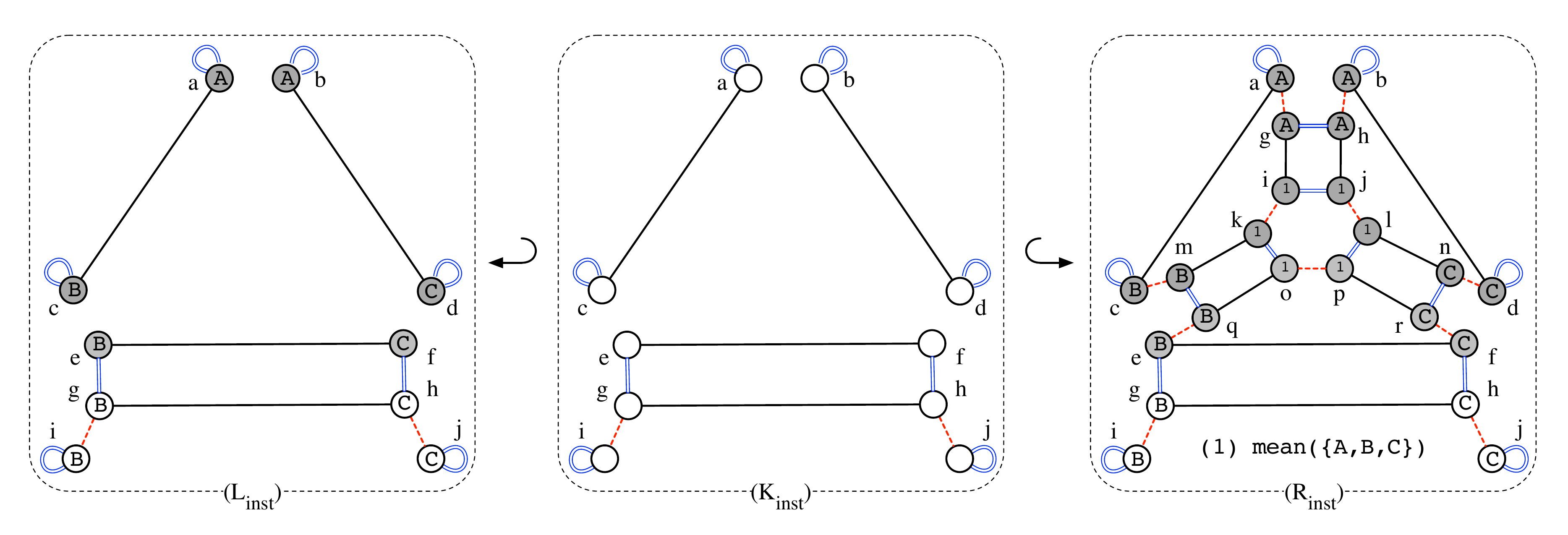}}
  \caption{Triangulation scheme}
      \vspace{-0.3cm}
    \label{fig:triang_scheme}
 \end{figure}

\begin{definition}[Graph scheme]
Let $G$ be a $\Pi$-embedded $n$-G-map. Let us consider an embedding signature $\Sigma_\Pi$ and its corresponding embedding algebra ${\cal A}_{G}$.

A graph scheme $H$ on ${\cal T}_\Pi(V)$ is a $\Pi$-labelled graph on terms of ${\cal T}_\Pi(V)$.

Let $\sigma:V \rightarrow V_G$ be an interpretation of the variables, the evaluation $eval_\sigma(H)$ of the graph $H$ is the $\Pi$-labelled graph that has the same base ($eval_\sigma(H)_\bot=H_\bot$) such as for each embedding operation $\pi$ of $\Pi$, $l_{eval_\sigma(H),V,\pi} = l_{H,V,\pi} \circ eval_\sigma$.
\end{definition}

\begin{definition}[Rule scheme]
Let $\Pi$ be a set of embedding operations of dimension $n$ and $\Sigma_\Pi$ an embedding signature.

A rule scheme $r_{\cal T} : L_{\cal T} \hookleftarrow K_{\cal T} \hookrightarrow R_{\cal T}$ on $\Sigma_\Pi$ is defined by two inclusion morphisms $K_{\cal T} \hookrightarrow L_{\cal T}$ and $K_{\cal T} \hookrightarrow R_{\cal T}$ between the graph schemes $ L_{\cal T}$, $ K_{\cal T}$ and $ R_{\cal T}$ on ${\cal T}_\Pi(V_L)$ such that:
\begin{itemize}
\item node labels of $L_{\cal T}$ (and so, labels of $K_{\cal T}$) are undefined 
- {\em i.e.} $L_{\cal T} = \Pi_{\pi \in \Pi} (proj_\pi(L_{\cal T})_\bot)$;
\item 
$R_{\cal T}$ satisfies the embedding constraints of Definition~\ref{def:ebd_constr}.
\end{itemize}
\end{definition}

The instantiation mechanism of a rule scheme is  constructive and based on the match morphism $m : L_{\cal T} \rightarrow G$ between the left-hand side of the scheme rule $L_{\cal T}$ and the embedded G-map $G$ on which the rule schema is applied. The main underlying idea is basically to build from the considered pattern ($ L_{\cal T}$, $K_{\cal T}$ or $R_{\cal T}$) and from the match morphism $m$, a graph completed with all nodes (and arcs) belonging to orbits whose embedding values can potentially be modified by the application of the rule. The resulting graphs are respectively denoted as $L[m]$, $K[m]$ and $R[m]$.
\begin{itemize}
\item the left hand-side $L[m]$ of the instantiated rule will consist of all matched nodes  together with nodes whose embedding values can be indirectly modified and of all associated embedding values. 
\item similarly, the kernel $K[m]$ will be built following the same construction, but without node labels.
\item the right hand-side $R[m]$ will include $K[m]$ and be completed with added parts and labels of $R_{\cal T}$ that are evaluated.
\end{itemize}

\begin{definition}[Rule scheme instantiation]
Let $\Pi$ be a set of embedding of dimension $n$ and $\Sigma_\Pi$ an embedding signature.
Let $r_{\cal T} : L_{\cal T} \hookleftarrow K_{\cal T} \hookrightarrow R_{\cal T}$ be a rule scheme on  $\Sigma_\Pi$, $m : L_{\cal T} \rightarrow G$ be a match morphism on a $\Pi$-embedded $n$-G-map $G$, and ${\cal A}_{G}$ be a $\Sigma_\Pi$-embedding algebra.

The instantiated rule $r[m] : L[m] \hookleftarrow K[m] \hookrightarrow R[m]$ is defined by  

\begin{itemize}
\item $L[m]=Lsat_{\Pi \times V_{K_{\cal T}}}(L_{\cal T})$, 
\item $K[m]=Ksat_{\Pi \times V_{K_{\cal T}}}(K_{\cal T})$, 
\item and $R[m]=Rsat_{\Pi \times V_{K_{\cal T}}}(R_{\cal T})$,
\end{itemize}
where the saturation operators  $Lsat$, $Ksat$ and $Rsat$ are recursively defined on $\Pi \times V_{K_{\cal T}}$.

Let us define the saturation operators   $Lsat$, $Ksat$ and $Rsat$  by the following induction principle over the elements of the  set  $\Pi \times V_{K_{\cal T}}$ :

\begin{list}{\labelitemi}{\leftmargin=0.5cm}

\item{\bf base case} $\Pi \times V_{K_{\cal T}} = \emptyset$.

Let $\sigma_m:V_{L_{\cal T}} \rightarrow V_G$ be the substitution that associates to each node $v$ of $L_{\cal T}$ its image $m(v)$ along the match morphism $m$. 

$Lsat_{\emptyset}(L_{\cal T})$, $Ksat_{\emptyset}(K_{\cal T})$ and $Rsat_{\emptyset}(R_{\cal T})$ are the graphs respectively isomorphic to $m(L_{\cal T})$  (the node images with all their embedding values and arcs issued from $L_{\cal T}$), $Prod_{\pi \in \Pi} (proj_\pi(m(K_{\cal T}))_\bot)$ and $eval_{\sigma_m}(R_{\cal T})$ such that the following inclusions exist:
$
Lsat_{\emptyset}(L_{\cal T}) \hookleftarrow Ksat_{\emptyset}(K_{\cal T}) \hookrightarrow Rsat_{\emptyset}(R_{\cal T})
$.

Let $h_{Lsat_{\emptyset}}:L_{\cal T} \rightarrow Lsat_{\emptyset}(L_{\cal T})$ be the morphism that associates each node $v$ of $L_{\cal T}$ to the node of $Lsat_{\emptyset}$ isomorphic to $m(v)$.

Let $g_{Lsat_{\emptyset}}:Lsat_{\emptyset}(L_{\cal T}) \rightarrow G$ that associates each node of $Lsat_{\emptyset}(L_{\cal T})$ that is isomorphic to $m(v)$ to $m(v)$ itself.
In particular, for all node $v$ of $L_{\cal T}$, $g_{Lsat_{\emptyset}}(h_{Lsat_{\emptyset}}(v))=m(v)$.

\item{\bf induction step} $\Pi \times V_{K_{\cal T}} \neq \emptyset$

Let note a subset $PV \subset \Pi \times V_{K_{\cal T}}$, a $\Pi$-labelled rule $Lsat_{PV}(L_{\cal T}) \hookleftarrow Ksat_{PV}(K_{\cal T}) \hookrightarrow Rsat_{PV}(R_{\cal T})$, and two morphisms $h_{Lsat_{PV}}:L_{\cal T} \rightarrow Lsat_{PV}(L_{\cal T})$ and $g_{Lsat_{PV}}:Lsat_{PV}(L_{\cal T}) \rightarrow G$.

Let $\pi:<\!\!o\!\!> \rightarrow s \in \Pi$ and $v \in V_{K_{\cal T}}$ with $(\pi,v) \not\in PV$.

Let us construct $Lsat_{\{(\pi,v)\} \cup  PV}(L_{\cal T})$ with the appropriate morphisms.
   
\begin{center}~
\xymatrix{
    ~ & \ar@{->}[ld]_b  Lsat_{PV}(L_{\cal T})<\!\!o\!\!>(h_{Lsat_{PV}}(v)) \ar@{^{(}->}[rd]^a & ~ & 
	    \ar@{->}[ld]_{h_{Lsat_{PV}}} L_{\cal T} \\
(*)
	\ar@{->}[rd]^d & (1) & \ar@{->}[ld]_c
	 Lsat_{PV}(L_{\cal T}) \ar@{->}[rd]^{g_{Lsat_{PV}}}  & ~ & ~\\
	~ & Lsat_{\{(\pi,v)\}\cup PV}(L_{\cal T}) 
		 \ar@{->}[rr]_{g_{Lsat_{\{(\pi,v)\}\cup PV}}}^*+{~~~~~~~~~~~~(2)} 
		 & ~ & G & ~ \\
}\\
	$\mbox{where} (*) proj_\pi(\mbox{G<\!\!o\!\!>}(m(v))) \times Prod_{\pi' \in \pi \backslash \pi} proj_{\pi'}(\mbox{G<\!\!o\!\!>}(m(v)))_\bot $
    \end{center}

Let us define the morphisms
\begin{itemize}
\item $a: Lsat_{PV}(L_{\cal T})<\!\!o\!\!> (h_{Lsat_{PV}}(v)) \hookrightarrow Lsat_{PV}(L_{\cal T})$ 
\item 
and  $b: Lsat_{PV}(L_{\cal T})<\!\!o\!\!> (h_{Lsat_{PV}}(v)) \rightarrow proj_\pi(\mbox{G<\!\!o\!\!>}(m(v))) \times Prod_{\pi' \in \Pi \backslash \pi} proj_{\pi'}(\mbox{G<\!\!o\!\!>}(m(v)))_\bot$ 
\end{itemize}
such that for all node or arc $x$ of $Lsat_{PV}(L_{\cal T})<\!\!o\!\!> (h_{Lsat_{PV}}(v))$, $b(x)=g_{Lsat_{PV}}(x)$.

Let us define $Lsat_{\{(\pi,v)\}\cup PV}(L_{\cal T})$ as the pushout (1) of $a$ and $b$ defined by 
\begin{itemize}
\item $c: Lsat_{PV}(L_{\cal T}) \rightarrow Lsat_{\{(\pi,v)\}\cup PV}(L_{\cal T})$ 
\item and $d: proj_\pi(\mbox{G<\!\!o\!\!>}(m(v))) \times Prod_{\pi' \in \Pi \backslash \pi} proj_{\pi'}(\mbox{G<\!\!o\!\!>}(m(v)))_\bot \rightarrow Lsat_{\{(\pi,v)\}\cup PV}(L_{\cal T})$.
\end{itemize}

Let $h_{Lsat_{\{(\pi,v)\}\cup PV}} = c \circ h_{Lsat_{PV}}$.

And let $g_{Lsat_{\{(\pi,v)\}\cup PV}}: Lsat_{\{(\pi,v)\}\cup PV} \rightarrow  G$ be the morphism such as diagram (2) is commutative and $g_{Lsat_{\{(\pi,v)\}\cup PV}} \circ d$ be the identity for all node or arc $x$ of $proj_\pi(\mbox{G<\!\!o\!\!>}(m(v))) \times Prod_{\pi' \in \Pi \backslash \pi} proj_{\pi'}(\mbox{G<\!\!o\!\!>}(m(v)))_\bot$ (that is always possible because (1) and (2) are commutative).

In particular, for all node $v$ of $L_{\cal T}$, $g_{Lsat_{\{(\pi,v)\}\cup PV}}(h_{Lsat_{\{(\pi,v)\}\cup PV}}(v))=m(v)$ because (2) commutes and the induction hypothesis on $Lsat_{PV}$ and its associated morphism.

The construction of $Ksat_{\Pi,V_{K_{\cal T}}}(K_{\cal T})$ and $Rsat_{\Pi,V_{K_{\cal T}}}(R_{\cal T})$  is similar with a difference in the labelling along $b$ and $d$.
For the kernel, as we want no label, we use $Prod_{\pi \in \Pi} proj_{\pi}(\mbox{G<\!\!o\!\!>}(m(v)))_\bot$ instead of $proj_\pi(\mbox{G<\!\!o\!\!>}(m(v))) \times Prod_{\pi' \in \Pi \backslash \pi} proj_{\pi'}(\mbox{G<\!\!o\!\!>}(m(v)))_\bot$. In the same way, for the right hand-side, as we want expression interpretations as node labels, we use $(proj_\pi(\mbox{G<\!\!o\!\!>}(m(v))))_{eval_{\sigma m}(l_{R_{\cal T},V,\pi}(v))} \times Prod_{\pi' \in \Pi \backslash \pi} proj_{\pi'}(\mbox{G<\!\!o\!\!>}(m(v)))_\bot$ instead of $proj_\pi(\mbox{G<\!\!o\!\!>}(m(v))) \times Prod_{\pi' \in \Pi \backslash \pi} proj_{\pi'}(\mbox{G<\!\!o\!\!>}(m(v)))_\bot$. 

The following inclusions hold:
$
Lsat_{\{(\pi,v)\}\cup PV}(L_{\cal T}) \hookleftarrow Ksat_{\{(\pi,v)\}\cup PV}(K_{\cal T}) \hookrightarrow Rsat_{\{(\pi,v)\}\cup PV}(R_{\cal T})
$.

Finally, the result match morphism $m^*:L[m] \rightarrow G$ is $m^*=g_{Lsat_{\Pi \times V_{K_{\cal T}}}}$.

\end{list}
\end{definition}

Let us note that the inclusion morphism $a$ always exists, by definition of orbits. For the left-hand and kernel parts, it is clear that $b$ exists, since all added graphs during  saturation are included in $G$. For the right-hand part, existence of $b$ depends on the condition imposed on $R_{\cal T}$ by the rule scheme definition. Thus, the saturation with $(\pi,v)$ and $(\pi,w)$ for two nodes that belong to the same orbit ({\em ie} $v \equiv_{R_{\cal T}<\!\!o\!\!>}w$ where $<\!\!o\!\!>$ is the domain of $\pi$) adds the same graph. Especially, $l_{R_{\cal T}, V, \pi}(v) = l_{R_{\cal T}, V, \pi}(w)$, and thus \\
$proj_\pi(\mbox{G<\!\!o\!\!>}(m(v)))_{eval_{\sigma m}(l_{R_{\cal T},V,\pi}(v))} = proj_\pi(\mbox{G<\!\!o\!\!>}(m(w)))_{eval_{\sigma m}(l_{R_{\cal T},V,\pi}(w))}$.

At each saturation step, graphs added to the left-hand side, to the kernel, and to the right-hand side have the same base. Thus, the double inclusion always exists with an adequate choice of node names and arc names.

The saturation order of $(\pi,v)$ couples does not matter, because the construction of the morphism $b$ guarantees an unique addition of nodes and arcs of G (or their isomorphisms) to the instantiated rule.

\begin{theorem}[preservation of embedded G-map's consistency]
\label{theo:ebdGmap_preserv}
Let $\Pi$ be a set of embedding of dimension $n$, 
$r_{\cal T} : L_{\cal T} \hookleftarrow K_{\cal T} \hookrightarrow R_{\cal T}$ be a rule scheme and  $m : L_{\cal T} \rightarrow G$ be a match morphism on a $\Pi$-embedded $n$-G-map $G$. If $r_{\cal T}$ satisfies the conditions of topological consistency preservation for $m$, the direct transformation $G \Rightarrow^{r[m],m^*} H$ with the instantiated rule $r[m]$ exists and produces a $\Pi$-embedded $n$-G-map $H$.
\end{theorem}

\begin{proof} The proof of this theorem can be found in the technical report~\cite{rapport}.\end{proof}

   \vspace{-0.5cm}
\section*{Conclusion}

In this article, in the context of topology-based geometric modelling, we have proposed a representation of  embedded $n$-dimensional objects as a particular class of $I$-labelled graphs. Nodes have as many labels as there are different kinds of data to represent the geometric embedding.
The category of $I$-labelled graphs is defined as a natural extension of the partially labelled graphs defined in \cite{Habel-Plump02}. Considering the modelling operations, we extend a rule-based language \cite{PACL2008_2394} used to define topological operations. We introduce embedding variables and expressions on rule node labels to deal with the computation of the embedding of constructed objects. The resulting language allows to define geometric operations in an easy and safe way, as constraints on rules ensure both topological and geometric consistency.

Moreover, we have already designed a first prototype of a topology-based geometric modeler, but only for pure topological operations described with rule schemes based on topological variables~\cite{Poudret09}. As previously mentioned, these variables allow us to define topological operations independently from the size of cells, that is, from the number of nodes constituting the cell to be filtered. For example, it allows us to define the topological triangulation of a triangle, a square, or any face with a single generic rule. The tool can be seen as a rule-application engine dedicated to our topological transformation rules~\cite{smi}. It allows us to quickly design and implement a modeler by specifying both its topological dimension and its set of application dedicated rules.  For usual topological operations, the prototype efficiency is comparable to  other topology-based geometric modelers based on G-maps. An unquestionable benefit of our approach is that topological operations can be quickly designed and implemented and that prototyped modelers are easily and safely extensible~\cite{smi}.  We are now extending this first prototype with embedding variables to deals with geometric operations. The combination of the two kind of variables has still to be formalized but the first developments attest of their compatibility.

   \vspace{-0.5cm}
\bibliographystyle{eptcs}
\bibliography{biblio}

\end{document}